\documentclass[11pt]{article}

\usepackage[letterpaper,margin=1in]{geometry}
\usepackage[T1]{fontenc}
\usepackage[utf8]{inputenc}
\usepackage{bm}
\usepackage{graphicx} % Required for inserting images
\usepackage{amsmath}
\usepackage{amssymb}
\usepackage{amsfonts}
\usepackage{textcomp}
\usepackage{xcolor}
\usepackage{booktabs}
\usepackage{soul}
\usepackage{cancel}
\usepackage{quoting}
\usepackage[caption=false]{subfig}
\usepackage{amsthm}
\usepackage{listings}
\usepackage[hidelinks]{hyperref}

\allowdisplaybreaks % from amsmath package, allows multiline equations to break across pages (delete if not wanted)
\newtheorem{lemma}{Lemma}
\newcommand{\be}{\begin{equation}}
\newcommand{\ee}{\end{equation}}
\def\be#1\ee{\begin{equation}#1\end{equation}}
\newcommand{\beaw}{\begin{eqnarray} }
\newcommand{\eeaw}{\end{eqnarray}}
\def\beaw#1\eeaw{\begin{align}#1\end{align}}
\newcommand{\bea}{\begin{eqnarray*}}
\newcommand{\eea}{\end{eqnarray*}}

\newcommand{\tbt}[4]{
  \left[ \renewcommand{\arraystretch}{0.7}\begin{array}{cc}
       #1 & #2 \\ #3 & #4
   \end{array} \right] }

\newcommand{\tbo}[2]{
  \left[ \renewcommand{\arraystretch}{0.8} \begin{array}{c}
       #1 \\ #2
       \end{array} \right] }
\newcommand{\obt}[2]{
  \left[ \begin{array}{cc}
       #1 & #2
       \end{array} \right] }
		
\newcommand{\TbT}[9]{
  \left[ \renewcommand{\arraystretch}{.8} \begin{array}{ccc}
       #1 & #2 & #3 \\
		   #4 & #5 & #6 \\
			 #7 & #8 & #9
       \end{array} \right]}					
			
\newcommand{\Tbo}[3]{ \left[ \renewcommand{\arraystretch}{.8} \begin{array}{c} #1 \\ #2 \\ #3 \end{array} \right] }
\newcommand{\Fbo}[4]{ \left[ \begin{array}{c} #1 \\ [-0.35em] #2 \\ [-0.35em] #3 \\ [-0.35em] #4 \end{array} \right] }

	\newcommand{\ObF}[4]{
  \left[ \begin{array}{cccc}
       #1 & #2 & #3 & #4
       \end{array} \right] }

\newcommand{\TBT}[6]{
  \left[ 
	\renewcommand{\arraystretch}{1.15}
	\begin{array}{ccc}
       #1 & #2 & #3 \\
		   #4 & #5 & #6 
       \end{array} \right]}

\newcommand{\del}[1]{\ifmmode\cancel{#1}\else\st{#1}\fi}
\definecolor{codegray}{rgb}{0.96, 0.96, 0.96}
\definecolor{commentgreen}{rgb}{0.0, 0.6, 0.0}
\definecolor{keywordblue}{rgb}{0.0, 0.0, 0.8}
\definecolor{stringpurple}{rgb}{0.58, 0.0, 0.82}

\hypersetup{%
	pdfauthor={David J. N. Limebeer and Charl van de Merwe},
	pdftitle={The Flat Earth Error: Differential Geometry in Vehicle Dynamics},
	pdfkeywords={curvature, differential geometry, kinematic constraints, vehicle dynamics, optimal control},
	pdfsubject={Differential geometry in vehicle dynamics}
}

\title{The Flat Earth Error: Differential Geometry in Vehicle Dynamics}

\author{%
David J. N.\ Limebeer\thanks{\texttt{david.limebeer@wits.ac.za}}\\
Department of Electrical and Information Engineering\\
University of the Witwatersrand\\
Johannesburg, South Africa
\and
Charl van de Merwe\thanks{\texttt{charlvandemerwe1@students.wits.ac.za}}\\
Department of Electrical and Information Engineering\\
University of the Witwatersrand\\
Johannesburg, South Africa
}

\date{}

\begin{document}

\maketitle

\begin{abstract}
The widespread idealization of road surfaces as horizontal planes can introduce significant inaccuracies into vehicle dynamics simulations, a phenomenon termed the ``Flat Earth Error.'' This article provides an expository guide to the use of classical differential geometry to model vehicular motion on curved surfaces. The influence of road curvature is characterized using the metric tensor, the second fundamental form, the shape operator, and the Christoffel symbols. Reproducible MATLAB scripts using an elliptic cone as a benchmark example illustrate the construction of geodesic curves and the simulation of particle dynamics on curved surfaces. The resulting geometric structures are integrated into a single-track vehicle-and-track model within a pseudospectral optimal control framework. Trajectory optimization results over a high-density mobile LiDAR profile of Darlington Raceway are used to generate a high-fidelity road-surface model. This model is used within an hp-adaptive collocation framework to investigate minimum lap time optimal control vehicular trajectories. These computations capture the non-smooth traction saturation limits of the tyres alongside position-dependent variations in gravitational forcing. Integrating differential geometry, multibody mechanics, and optimal control is essential for high-fidelity driven and autonomous vehicle-dynamics simulations. The optimized velocity and tyre slip profiles show that real-world racing track geometries induce dynamically significant three-dimensional effects. These results are particularly relevant to performance-limited simulations on highly banked track surfaces such as NASCAR ovals.
\end{abstract}


\section*{Nomenclature} \label{sec:Nomenclature}

\begin{tabbing}
\hspace*{1.8cm} \= \kill % Sets the column width for the symbols

$a, b$ \> Elliptic cone semi-axes coordinates (m) \\
$c$ \> Cone slope parameter (dimensionless) \\
$E, F, G$ \> Components of the First Fundamental Form (dimensionless) \\
$F_s, F_n$ \> Generalized forces acting along curvilinear track coordinates (N) \\
${\bm g}$ \> Acceleration due to gravity ($\text{m/s}^2$); bold type is used for vector quantities \\
$g_c$ \> The gravity scalar constant $|\bm g|$ ($\text{m/s}^2$) \\
$g_{ij}$ \> Covariant components of the surface metric tensor (dimensionless) \\
$h$ \> Height of the vehicle's mass centre above the road (m) \\
$I_y$ \> Car's pitch moment of inertia ($\text{kg}\cdot\text{m}^2$) \\
$I_z$ \> Car's yaw moment of inertia ($\text{kg}\cdot\text{m}^2$) \\
$m$ \> Mass of the vehicle (kg) \\
$L, M, N$ \> Components of the matrix in the second fundamental form (\text{m}$^{-1}$)\\
$q^i$ \> Generalized surface manifold coordinates $q^1$ and $q^2$ (m) \\
${\bm p}$ \> Position vector in 3D space (m) \\
$s, n$ \> Local orthogonal curvilinear track coordinates along the road surface (m) \\
$U, V$ \> Isometric Cartesian coordinates in the 2D developed flat plane (m) \\
$z$ \> Extrinsic vertical elevation coordinate in Euclidean space (m) \\
\\
\textbf{Greek Letters} \\
$\alpha_f, \alpha_r$ \> Front and rear tyre slip angles (rad) \\
$\delta$ \> Vehicle steering angle (rad) \\
$\Gamma^k_{ij}$ \> Christoffel symbols of the second kind (connection coefficients) (\text{m}$^{-1}$) \\
$\theta$ \> Physical azimuthal tracking angle on the cone manifold (rad) \\
$\phi$ \> Unfolded sector configuration angle in the 2D developed plane (rad) \\
$\phi_{\text{total}}$ \> Total isometric angular span of the developed cone sector (rad) \\
$\chi$ \> Heading/slip angle between $\mathbf{e}_x^b$ and $\partial_s\mathbf{p}$ (rad)
\end{tabbing}

\section{Introduction}
\label{sec:intro}

\subsection{Historical Background} \label{subsec:historical_lineage}
The geometry of the Earth's surface was established in antiquity, with a spherical Earth proposed by Pythagoras. The first physical proof was based on Aristotle's observation of the Earth's shadow on the Moon. Circa 250\,BC, and with remarkable accuracy, Eratosthenes calculated the Earth's radius using the angles of sun shadows in different cities. Copernicus proposed the heliocentric model in the sixteenth century and tactically linked his opponents' arguments to the Flat-Earth theory as a way to discredit them~\cite{Russel1991}. 
The seed of the ``Flat-Earth error'' was thus planted, but it did not grow to choke the truth until the publication of Washington Irving's fictionalized best-selling biography of Columbus~\cite{Irving1828}. The Flat-Earth myth gained further momentum towards the end of the nineteenth century when Flammarion's flat Earth engraving appeared in his book \emph{L'Atmosph\`ere:\,M\'et\'eorologie}; see Figure\,\ref{fig:map}.
\begin{figure}[ht] 
    \begin{center}
        \includegraphics[width=0.6\textwidth]{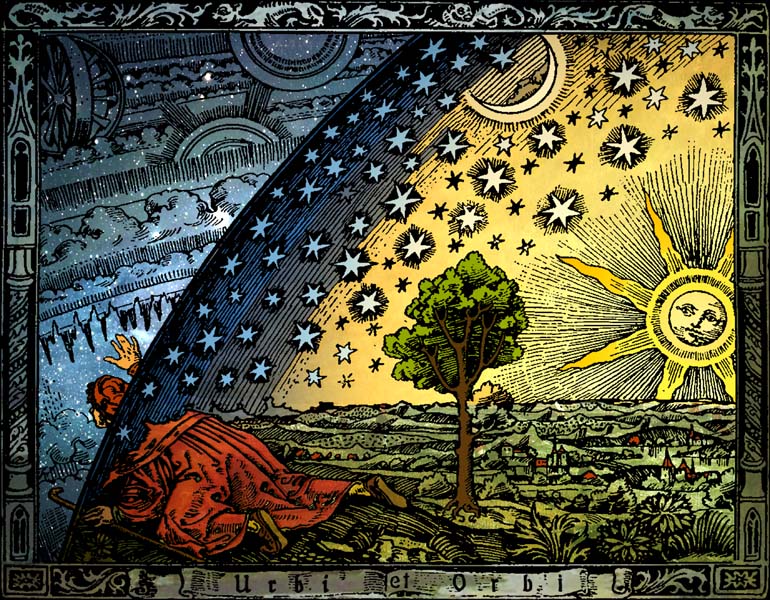}
        \caption{The point where the sky and the Earth touch (from Flammarion's 1888 engraving).}
        \label{fig:map}
    \end{center}
\end{figure}
The Flat-Earth myth has proven remarkably resilient, surviving as a ``fact'' in school textbooks and popular culture long after it had been debunked by historians.

\subsection{The Flat-Earth Fraternity and Vehicle Simulation}
For various reasons, the creators of many vehicle simulation models still find common cause with the Flat-Earth fraternity. In these models, the vehicle's motion is constrained to a horizontal plane with gravity always 
perpendicular to it. 
In Flat-Earth models the vehicle's potential energy is deemed constant, while undulating (curved) surfaces produce a dynamic exchange between the vehicle's kinetic and potential energy. Other influences include 
the ever-changing direction of gravity, the reduction in the tyres' force-generating capacity as a vehicle negotiates the summit of a hill, or their increased capability when a vehicle drives through a dip. Centripetal forces act towards the centre of curvature and perpendicular to the vehicle's velocity;
this is a 3D phenomenon. Similar effects occur with cambered corners, with adverse camber particularly troublesome at high speed. These differential-geometric influences make curved-road-surface vehicle models both more complex and more interesting. 

Flat-Earth road models are simple, well-understood, computationally efficient, and sufficient for many driven and autonomous vehicle applications. Early flirtations with 3D road surfaces focus on straight, undulating roads~\cite{Roos1997}. It was recently observed that vehicle models built on a level terrain assumption are inadequate for the capture of rolling or pitching dynamics---the roll-over of such vehicles is a real practical risk~\cite{Badar2024}. A reasonably up-to-date survey of road modelling is presented in~\cite{Limebeer2023}, with early signs of the full power of differential geometry becoming evident in~\cite{Fork2021,LimebeerVSD2021,Fork2024,Limebeer2025}. The purpose of this paper is to accelerate the adoption of ideas from classical differential geometry that are expertly described in several books including~\cite{DoCarmo1976,Shifrin2006}.

\subsection{The origins of curvature}
It would be wrong to say that the Flat-Earth error is ubiquitous in the vehicle dynamics community, but it is nonetheless a widely adopted idealisation that is arguably overused. In order to disaggregate some of the issues related to road modelling, we consider Figure\,\ref{fig:earth} which shows the Earth with a plane ${\mathcal T}_{\mathcal M}^P$ tangent to the surface at point ${\bm p}(s,n)$.
\begin{figure}[ht] 
    \begin{center}
        \includegraphics[trim=40mm 160mm 40mm 50mm, clip, width=0.6\textwidth]{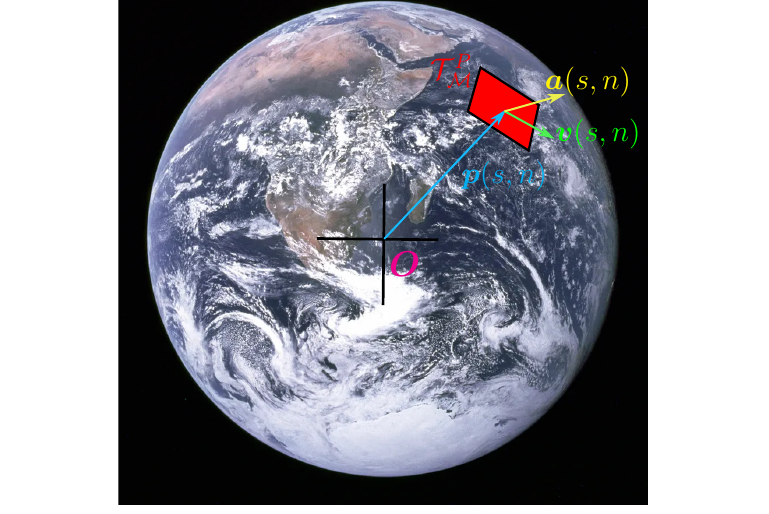}
        \caption{The Earth with tangent plane ${\mathcal T}_{\mathcal M}^P$ at position ${\bm p}(s,n)$, together with velocity ${\bm v}(s,n)$ and acceleration ${\bm a}(s,n)$ vectors. The origin of the coordinate system is ${\bm O}$.}
        \label{fig:earth}
    \end{center}
\end{figure}
This localized path position point is given by:
\be
    {\bm p}(s,\,n) = x(s,n){\bm e}_x + y(s,n){\bm e}_y + z(s,n){\bm e}_z \label{eq:PP}
\ee
in Cartesian coordinates. Although the Earth's surface is three dimensional, ${\bm p}(s,n)$ is confined to the Earth's surface and so only two independent variables are required to describe it (surfaces in ${\mathbb R}^3$ are intrinsically two-dimensional). As an example, $s$ may be defined as the longitude, and $n$ the latitude of ${\bm p}(s,\,n)$.
The Cartesian frame has origin ${\bm O}$ and Earth-fixed basis vectors ${\bm e}_x$, ${\bm e}_y$, and ${\bm e}_z$. Equation~\eqref{eq:PP} is an extrinsic description of ${\bm p}(s,\,n)$, because reference is made to the embedded space $\mathbb{R}^3$. It follows from the chain rule that the velocity of ${\bm p}(s,n)$ is given by:
\begin{equation}
    {\bm v}(s,n) = \dot{\bm p}(s,n) = \dot{s} \partial_s {\bm p}(s,n) + \dot{n} \partial_n {\bm p}(s,n) \label{eq:vel}
\end{equation}
where the vector partial derivatives $\partial_s {\bm p}(s,n)$ and $\partial_n {\bm p}(s,n)$ are assumed independent. In general, the velocity vector ${\bm v}(s,n)$ is in the tangent plane ${\mathcal T}_{\mathcal M}^P$, but not in the manifold ${\mathcal M}$ (the Earth's surface). The acceleration vector ${\bm a}(s,n)$ is typically in neither. 

The partial derivative operators $\partial_s (\cdot)$ and $\partial_n (\cdot)$, which make no reference to ${\bm p}(s,n)$, can be considered vectors in ${\mathcal T}_{\mathcal M}^P$ with intrinsic descriptions. Intrinsic descriptions of this type are required when there is no external coordinate system and no point of origin (such as the universe studied in General Relativity). In early antiquity, an observer located at ${\bm p}(s,n)$ would consider the Earth's surface synonymous with the tangent plane ${\mathcal T}_{\mathcal M}^P$---there was no notion of an underlying spherical manifold $\mathcal M$.

\begin{quotation}
    The question is: ``Can one assume that the road surface is tangent to the Earth's surface, so that it can be treated as a resident of ${\mathcal{T}}_{\mathcal M}^P$ with gravity always acting in the normal direction, or is the Earth's surface ``wrinkled'' so that ${\mathcal{T}}_{\mathcal M}^P$ is position-dependent, far smaller, and only supportive of the vehicle's wheels, for example, and with the direction of gravity forever changing?'' This difference is the essence of the ``Flat-Earth'' error.
\end{quotation}

\subsection{Outline and Document Organization}
We will provide an account of the differential geometry underlying the second possibility. 
Section\,\ref{sec:Curved} introduces the mathematical tools needed to measure lengths and angles on curved surfaces, and provides a precise characterization of curvature.
A method for decomposing acceleration into tangential and normal components is provided in Section\,\ref{sec:Gauss}.
Geodesic curves are also described, which are curves along which a vehicle will travel if lateral tyre forces are not allowed---as if the vehicle is travelling on ice and without the influence of gravity. Geodesics are analogous to straight lines in the plane. 
The geodesics of the elliptic cone are analysed in Section\,\ref{sec:Elliptic},
while the dynamics of a unit-mass particle is studied in Section\,\ref{sec:Dynamics}.
Elliptic conical track surfaces have been used in a vehicular context in~\cite{Limebeer2025B}; the interior surface of elliptic cones is reminiscent of some NASCAR ovals. Section\,\ref{sec:VehicleDynamics} demonstrates how the geometric framework developed in this paper can be embedded in practical vehicle dynamics simulations.
A simple single-track vehicle model is described in the Section\,\ref{sec:singletrack}, which illustrates how road curvature may be introduced into an integrated system model.
A representation vehicular optimal control problem (OCP) is described in Section\,\ref{sec:OCprob}, and conclusions drawn in Section\,\ref{sec:conclusions}.
The appendices contain computer codes, and vehicle, tyre and numerical optimal control setup parameters.
	
\section{Curved surfaces} \label{sec:Curved}
In order to highlight the role of `surface curvature' in dynamics problems, the dynamics of a unit mass on a curved surface will be considered,
which is an exemplar problem in physics \cite{Villarreal2018}.
The key concepts, which will now be introduced, are the {\em metric tensor} or {\em first fundamental form}, the {\em second fundamental form}, and the {\em shape pencil} with its {\em principal curvatures} and {\em principal curvature directions}, and the {\em Christoffel symbols}.

\subsection{Measuring distance}
In standard three-dimensional Euclidean space, measuring the distance $dl$ between two points is done using Pythagoras' theorem; 
the associated distance is given by $dl^2 = dx^2+dy^2+dz^2$.
Implicit in this calculation is the understanding that `distance' is measured along straight lines
on a rectangular grid.
The {\em metric tensor} in this case is a $3 \times 3$ identity matrix ${g}_{ij}$ that allows us to write:
\be
dl^2 = \sum_{i=1}^3 \sum_{j=1}^3 {g}_{ij} dx^i dx^j
\ee
or, using the Einstein summation convention, $dl^2 = {g}_{ij} dx^i dx^j$,
in which the $dx^i$s represent $dx$, $dy$ and $dz$. Inner products are calculated as simply ${\bm u} \cdot {\bm v} = u^i v^i$.

A more interesting example comes from special relativity, where distances between events in spacetime are measured using the Minkowski metric (in the mostly negative sign convention): 
\[
dl^2 = \ObF{dt}{dx}{dy}{dz}
\left[ \begin{array}{cccc}
       c^2 & 0 & 0 & 0 \\[-0.35em]
		   0 & -1 & 0 & 0 \\[-0.35em]
			 0 & 0 & -1 & 0 \\[-0.35em]
			 0 & 0 & 0 & -1
       \end{array} \right]
\Fbo{dt}{dx}{dy}{dz} = {g}_{ij} dx^i dx^j,
\]
in which a time term $dt$ is added to the (negated) Cartesian distance. This switch from the positive-definite Euclidean geometry to a sign indefinite metric is not just a mathematical oddity, it is the mechanism that embeds the speed of light as a universal cosmic speed limit.
For physically realisable trajectories in spacetime, involving massive objects, this length metric must be positive and so $(dt)^2(c^2 - (\frac{dx}{dt})^2 - (\frac{dy}{dt})^2 - (\frac{dz}{dt})^2) > 0$. 
When $dl^2 =0$ the separation path matches the trajectory of massless particles like light.
When $dl^2 <0$ the distance between two events is so far apart that no signal can travel fast enough to connect them.
In the general case, the 
inner product becomes ${\bm u} \cdot {\bm v} = {g}_{ij} du^i dv^j$, which is different from the more familiar ${\bm u} \cdot {\bm v} = u^i v^i$.

When surfaces are curved, the shortest distances between points are not normally measured along straight lines, but along geodesics (in the case of a sphere these are great circles\textemdash circles that share their centre with that of the sphere). Pythagoras' theorem no longer applies, and the metric tensor need not be diagonal. The role of the metric tensor is to define the inner product between vectors and, consequently, distances and angles.

As we have seen, the position of a point on the Earth's surface can be described by \eqref{eq:PP}, with the velocity of this point described by \eqref{eq:vel}. Both $\partial_s {\bm p}$ and $\partial_n {\bm p}$ are vectors that define the tangent plane ${\mathcal T}_{\mathcal M}^P$. To minimize notational clutter, the explicit dependence of these vectors on $s$ and $n$ is omitted from now on.
In this framework the point's motion is confined to the manifold ${\mathcal M}$,
and there is no need to introduce surface-constraining Lagrange multipliers in dynamical analyses.

From (\ref{eq:vel}), the magnitude squared of $\bm v$ is given by
\[
	\|{\bm v}\|^2_2 = \bm v \cdot \bm v = (\dot s \partial_s \bm p + \dot n \partial_n \bm p) \cdot (\dot s \partial_s \bm p + \dot n \partial_n \bm p)
\]	
which in matrix form, becomes
\begin{equation}\label{eq:KE}
\begin{aligned}
\|{\bm v}\|^2_2 &=
\begin{bmatrix}\dot s & \dot n\end{bmatrix}
\begin{bmatrix}
\partial_s {\bm p} \cdot \partial_s {\bm p} & \partial_s {\bm p} \cdot \partial_n {\bm p} \\
\partial_n {\bm p} \cdot \partial_s {\bm p} & \partial_n {\bm p} \cdot \partial_n {\bm p}
\end{bmatrix}
\begin{bmatrix}\dot s \\\dot n\end{bmatrix} \\
&=
\begin{bmatrix}\dot s & \dot n\end{bmatrix}
\underbrace{\begin{bmatrix} E & F \\ F & G \end{bmatrix}}_{I_P(\cdot,\cdot)}
\begin{bmatrix}\dot s \\\dot n\end{bmatrix}
\end{aligned}
\end{equation}
in which $I_P(\cdot,\cdot)$ is the metric tensor associated with the first fundamental form.
This first fundamental form is a bilinear form representing the inner product of two vectors in the tangent plane ${\mathcal T}_P{\mathcal M}$ with respect to the basis vectors $\partial_s {\bm p}$  and $\partial_n {\bm p}$. 
That is ${\bm u} \cdot {\bm v} = I_P({\bm u},{\bm v})$, while
$\|{\bm u}\|^2=I_P({\bm u},\,{\bm u})$ is a local measure of `distance' on the curved surface.

The distance travelled on surface $\mathcal M$ in time $T$ is given by
\be 
L = \int_0^T \sqrt{\obt{\dot s}{\dot n} I_P \tbo{\dot s}{\dot n}} dt.
\ee

For arbitrary basis vectors ${\bm e}_1$ and ${\bm e}_2$, we have
$E = {\bm e}_1 \cdot {\bm e}_1$, $F = {\bm e}_1 \cdot {\bm e}_2$, and $G = {\bm e}_2 \cdot {\bm e}_2$, with
\be
{\bm u} \cdot {\bm w} = \obt{u_1}{u_2}\tbt{E}{F}{F}{G} \tbo{w_1}{w_2} = g_{ij}u^iw^j,
\ee
so that
\be
\cos \gamma = \frac{{\bm u} \cdot {\bm w}}{
( {\bm u} \cdot {\bm u})^\frac{1}{2}
( {\bm w} \cdot {\bm w})^\frac{1}{2}} = \frac{g_{ij}u^iw^j}{\sqrt{g_{ij}u^iu^j}\sqrt{g_{ij}w^iw^j}}
\ee
is the angle between $\bm u$ and $\bm w$.
The metric tensor is not only about measuring distances\textemdash it's a fundamental intrinsic geometric structure that defines both lengths and angles measured on curved surfaces. In vehicular applications, the metric tensor can be determined using a mesh of road-surface measurements such as a LiDAR map \cite{Limebeer2025}.

\subsection{Measuring curvature}
We begin our discussion of curvature by considering the special case of the saddle given in Figure\,\ref{fig:saddle}, which is described in Monge coordinates by ${\bm p}(s,n)=[s,\,n,\,s^2 - n^2]$. The tangent plane ${\mathcal T}_{\mathcal M}^P$ at $\bm p$ is spanned by the partial derivatives 
$\partial_s {\bm p} = [1,\,0,\,2s]$ and $\partial_n {\bm p} = [0,\,1,\,-2n]$.
The unit normal to ${\mathcal T}_{\mathcal M}^P$ is given by
\be
{\mathcal N}_P
 = \frac{\partial_s {\bm p} \times \partial_n {\bm p}}{\| \partial_s {\bm p}
\times \partial_n {\bm p} \|}= \frac{(-2s,\,2n\,\,1)}{\sqrt{4s^2+4n^2+1}}.
\ee
The unit normal at the origin for the surface in Figure\,\ref{fig:saddle} is seen as a red arrow.
%The unit normal ${\mathcal N}_P$ is perpendicular to both $\partial_s {\bm p}$ and $\partial_n {\bm p}$.

\begin{figure}[ht] 
	\begin{center}
		\includegraphics[trim=50mm 85mm 50mm 90mm, clip, width=0.5\textwidth]{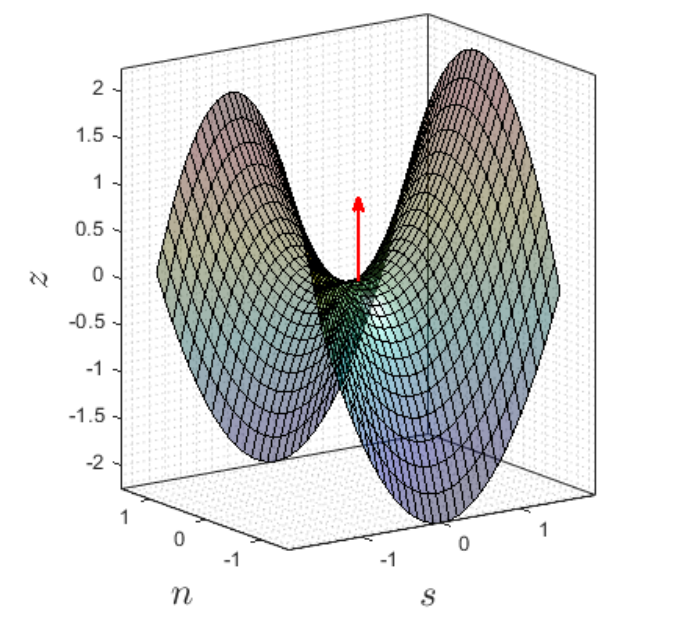}
		\put(-120,150) {${\mathcal N}_P$}
		\caption{Surface of the saddle $z = s^2 - n^2$ with a unit normal ${\mathcal N}_P$ depicted at the origin.}
		\label{fig:saddle}
	\end{center}
\end{figure}

As one moves along the line $(s,0)$ with $s$ increasing, the surface curves upwards, while along $(0,n)$ it curves downwards. This leads to several subsidiary observations: (i) the curvature of a surface is direction-dependent; 
(ii) curvature has a sign\textemdash the curvature of a sphere is deemed to be positive, while negative curvature is often associated with saddles; (iii) in one notion of curvature, the Gaussian curvature, zero curvature means that the surface is {\em developable}, or can be `flattened out'. For example, a cylinder or cone can be `unrolled' onto a plane without tearing or stretching; a sphere cannot be so developed.
It is also clear from the figure that as one moves around the surface, the normal ${\mathcal N}_P$ changes direction. If $\mathcal M$ is planar:
$
{\bm p}(s,n) = (s,\,n,\,as+bn),
$
and the surface normal ${\mathcal N}_P = (-a,\,-b,\,1)$ is independent of position.
Thus $\partial_s {\mathcal N}_P=0$ and $\partial_{n} {\mathcal N}_P=0$ in this particular case. 
On a plane, ${\mathcal N}_P$ is constant\textemdash this suggests that the rate of change of ${\mathcal N}_P$ can be used to measure curvature.  

An important idea in local surface theory is the {\em Gauss map}, which is a mapping from the ${\bm p}(s,n)$ induced surface to a surface induced by ${\mathcal N}_P$; the image of the Gauss map is a unit sphere. An example of the Gauss map for the saddle is shown in Figure\,\ref{fig:Gauss}.
Several facts are immediate: (i) Since ${\mathcal N}_P$ is normal to ${\mathcal T}_P{\mathcal M}$, it is also normal to the plane ${\mathcal T}_P{\mathcal G}$, which is tangent to the Gauss sphere. 
If the manifold ${\mathcal M}$ is planar, the point ${\bm n}_P$ on the Gauss sphere remains stationary wherever and whenever $\bm p$ moves.
If one moves in a small circle on ${\mathcal T}_P{\mathcal M}$, the normal associated with this movement will change direction tracing out an ellipse on the Gauss circle.
In some cases the direction of rotation around the ellipse on the Gauss map remains the same, sometimes the ellipse collapses into a line, or even a point, and sometimes it reverses direction.

\begin{figure*}[ht] 
	\begin{center}
		\includegraphics[trim=10mm 90mm 0mm 90mm, clip, width=.8\textwidth]{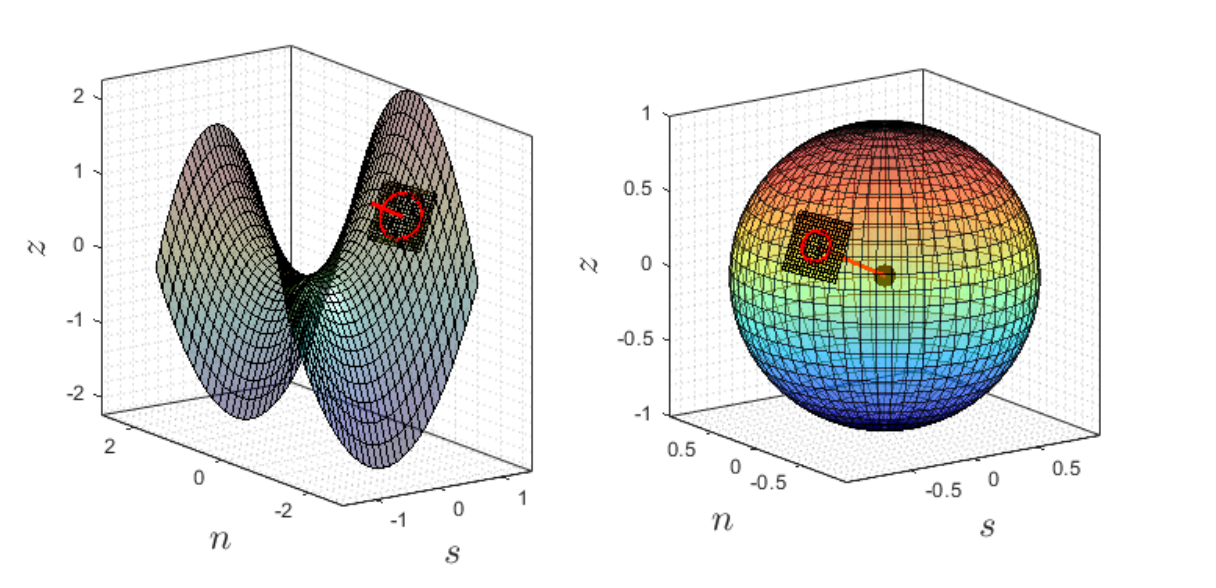}
		\caption{Surface of the saddle $z = s^2 - n^2$ with its corresponding Gauss map (a map from $(s,n)$ to the unit normal ${\mathcal N}_P$).}
		\label{fig:Gauss}
	\end{center}
\end{figure*}

Since ${\mathcal N}_P \cdot {\mathcal N}_P = 1$, it follows from the product rule that $\partial_s {\mathcal N}_P \cdot {\mathcal N}_P = 0$ and $\partial_n {\mathcal N}_P \cdot {\mathcal N}_P = 0$. 
Thus $\partial_s {\mathcal N}_P$ and $\partial_n {\mathcal N}_P$ must lie in ${\mathcal T}_{\mathcal M}^P$
and are expressible as linear combinations of $\partial_s {\bm p}$ and $\partial_n {\bm p}$. That is
\be
\partial _s {\mathcal N}_P= p \partial_s {\bm p} + q \partial_n {\bm p} \mbox{~~~~and~~~~} \partial _n {\mathcal N}_P= r \partial_s {\bm p} + s \partial_n {\bm p} \label{eq:Wein}
\ee

Differentiating $\partial_s {\bm p} \cdot {\mathcal N}_P=0$ and $\partial_n {\bm p} \cdot {\mathcal N}_P=0$ gives
\beaw
\partial_s(\partial_s {\bm p} \cdot {\mathcal N}_P) &= \partial_{ss}{\bm p} \cdot {\mathcal N}_P +  \partial_s {\bm p} \cdot \partial_s{\mathcal N}_P=0 ~\Rightarrow \partial_s {\bm p} \cdot \partial_s{\mathcal N}_P ~= -\partial_{ss}{\bm p}   \cdot {\mathcal N}_P \nonumber \\[-2pt] 
\partial_n(\partial_s {\bm p} \cdot {\mathcal N}_P) &= \partial_{ns} {\bm p}  \cdot {\mathcal N}_P +  \partial_s {\bm p} \cdot \partial_n{\mathcal N}_P=0 \,\Rightarrow \partial_s {\bm p} \cdot \partial_n{\mathcal N}_P \,= -\partial_{ns}  {\bm p} \cdot {\mathcal N}_P \nonumber \\[-2pt]
\partial_s(\partial_n {\bm p} \cdot {\mathcal N}_P) &= \partial_{sn} {\bm p}  \cdot {\mathcal N}_P +  \partial_n {\bm p} \cdot \partial_s{\mathcal N}_P=0 \,\Rightarrow \partial_n {\bm p} \cdot \partial_s{\mathcal N}_P \,= -\partial_{sn} {\bm p}  \cdot {\mathcal N}_P \nonumber \\[-2pt]
\partial_n(\partial_n {\bm p} \cdot {\mathcal N}_P) &= \partial_{nn} {\bm p}  \cdot {\mathcal N}_P +  \partial_n {\bm p} \cdot \partial_n{\mathcal N}_P=0 \Rightarrow \partial_n {\bm p} \cdot \partial_n{\mathcal N}_P = -\partial_{nn} {\bm p}  \cdot {\mathcal N}_P. \label{eq:2OPD}
\eeaw
The matrix of projections of the second-order derivative of $\bm p$ onto ${\mathcal N}_P$, is the
matrix associated with the Second Fundamental Form:
\be
\tbt{L}{M}{M}{N} = \tbt{\partial_{ss}{\bm p}   \cdot {\mathcal N}_P}{\partial_{sn}{\bm p}   \cdot {\mathcal N}_P}{\partial_{ns}{\bm p}   \cdot {\mathcal N}_P}{\partial_{nn}{\bm p}   \cdot {\mathcal N}_P}. \label{eq:SFF}
\ee
The second fundamental form $\Pi({\bm v},{\bm w})$ measures extrinsic curvature\textemdash how the surface curves away from its tangent plane.
Note: $\partial_n {\bm p} \cdot \partial_s {\mathcal N}_P = \partial_s {\bm p} \cdot \partial_n {\mathcal N}_P = -M$, by the symmetry of mixed partial derivatives.

One interpretation of the entries of the second fundamental form matrix is that they show how the surface normal bends away from ${\mathcal T}_P{\mathcal M}$ as $\bm p$ moves. 
To show this explicitly, suppose
\be
	\bm w = a\,\partial_s \bm p + b\,\partial_n \bm p, \qquad
	\bm v = c\,\partial_s \bm p + d\,\partial_n \bm p.
\ee
The directional derivative in the $\bm w$ direction is
\be
	\nabla_{\bm w} \mathcal N_P = a\,\partial_s \mathcal N_P + b\,\partial_n \mathcal N_P.
\ee
Thus
\beaw
	-(\nabla_{\bm w} \mathcal N_P)\cdot \bm v
	&= -\big(a\,\partial_s \mathcal N_P + b\,\partial_n \mathcal N_P\big)
	\cdot \big(c\,\partial_s \bm p + d\,\partial_n \bm p\big) \nonumber \\
	&= \obt{a}{b} \tbt{\partial_s {\bm p} \cdot \partial_s{\mathcal N}_P}
	{\partial_s {\bm p} \cdot \partial_n{\mathcal N}_P}
	{\partial_n {\bm p} \cdot \partial_s{\mathcal N}_P}
	{\partial_n {\bm p} \cdot \partial_n{\mathcal N}_P}  \tbo{c}{d}\nonumber \\
	&= \obt{a}{b} \tbt{L}{M}{M}{N} \tbo{c}{d} \label{eq:DirDer}
\eeaw
using \eqref{eq:2OPD} and \eqref{eq:SFF}.
To sum up so far, we began with a surface $\mathcal M$, with tangent vectors $\partial_s{\bm p}$ and $\partial_n {\bm p}$ spanning a plane ${\mathcal T}_{\mathcal M}^P$ (tangent to the surface at point $\bm p$), which has a
unit normal ${\mathcal N}_P$ (at $\bm p$). 
We then showed that the first-order derivatives of the normal also lie in the tangent plane, so that they are expressible as a linear combination of  $\partial_s{\bm p}$ and $\partial_n{\bm p}$.
Differentiating the orthogonality conditions $\partial_s {\bm p} \cdot {\mathcal N}_P=0$ and $\partial_n {\bm p} \cdot {\mathcal N}_P=0$ leads to the second fundamental form (with matrix coefficients $L$, $M$, and $N$). Some of the properties of surface curvature were outlined using intuitive arguments that we will now make concrete.

\subsection{The Weingarten equations and the shape pencil}
The Weingarten equations come from expressing $\partial_s{\mathcal N}_P$
and $\partial_n{\mathcal N}_P$ as linear combinations of $\partial_s{\bm p}$ and $\partial_n{\bm p}$; recall \eqref{eq:Wein}. 
Alternatively, recalling \eqref{eq:2OPD}, these equations come from
projecting the second-order partial derivatives $\partial_{ss}{\bm p}$, $\partial_{sn}{\bm p}$, and $\partial_{nn}{\bm p}$ onto ${\mathcal N}_P$.
That is
\beaw
\partial_{ss}{\bm p}   \cdot {\mathcal N}_P &= -(pE+qF) \nonumber \\
\partial_{sn}{\bm p}   \cdot {\mathcal N}_P &= -(pF+qG) \nonumber \\ 
\partial_{ns}{\bm p}   \cdot {\mathcal N}_P &= -(rE+sF) \nonumber \\
\partial_{nn}{\bm p}   \cdot {\mathcal N}_P &= -(rF+Gs) ,
\eeaw
which can be re-packaged in matrix form:
\be
\tbt{E}{F}{F}{G} \tbt{p}{r}{q}{s} = -\tbt{L}{M}{M}{N}.
\ee
The matrix
\be
S = \tbt{p}{r}{q}{s} = -I_P^{-1} \Pi_P \label{eq:Shape}
\ee
is a matrix representation of the {\em shape operator} or the {\em Weingarten map}.
The columns of $S$ give the tangent-plane coordinates of $\partial_s {\mathcal N}_P$
and $\partial_n {\mathcal N}_P$, which encode how the surface normal rotates
as one moves across the surface. 

In the case of the cylinder of radius $\rho$ shown in Figure\,\ref{fig:cylinder}, with $n$  measured circumferentially, and $s$ longitudinally,
the normal vector ${\mathcal N}_P$ does not change direction when moving in the $s$ direction. 
When moving in the $n$ direction, however, the rotational change of direction of $\partial_s {\mathcal N}_P$ is inversely proportional to the radius of the cylinder and the shape operator takes the form:
\[
\tbo{\partial_s {\mathcal N}_P}{\partial_n {\mathcal N}_P} = \tbt{0}{0}{0}{-\frac{1}{\rho}} \tbo{\partial_s {\bm p}}{\partial_n {\bm p}}.
\]
The negative sign is due to the fact that $\partial_n {\mathcal N}_P \propto -1/\rho$.
A sphere of radius $\rho$ has $S = (-1/\rho)I$, because the normal tilts
uniformly in every tangent direction. A general surface 
combines the two tangent directions that are dependent on the eigen-structure of $S$.

\begin{figure}[ht] 
	\begin{center}
		\includegraphics[trim=10mm 100mm 0mm 90mm, clip, width=0.6\textwidth]{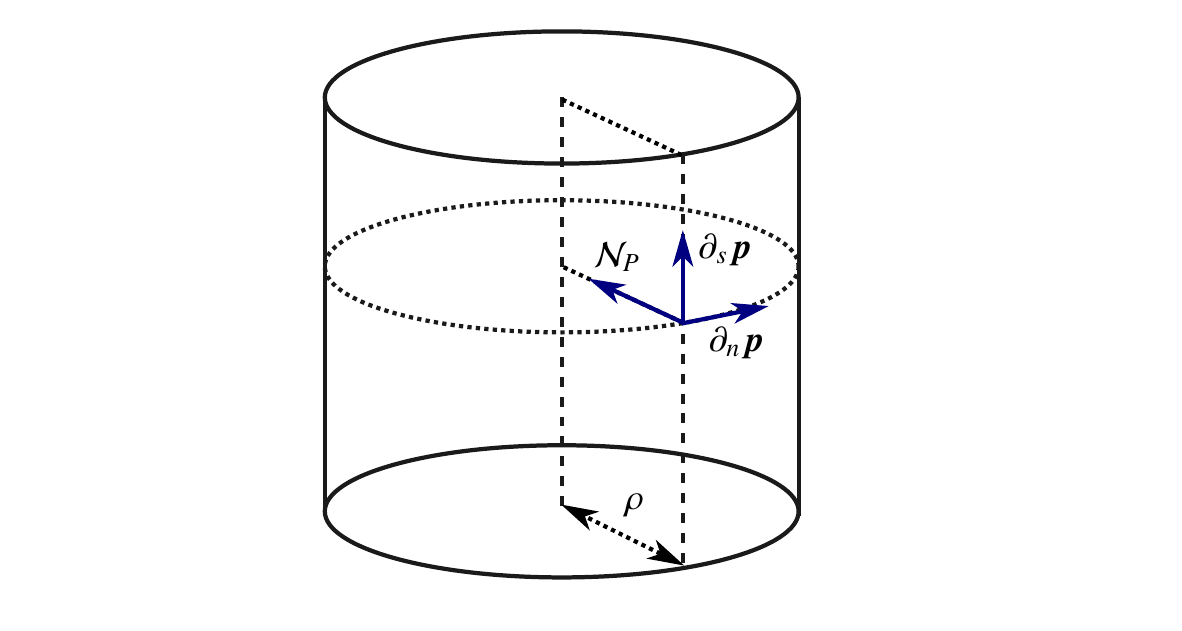}
		\caption{The normal changing direction with the circumnavigation of the surface of a cylinder.
			The variable $n$ is measured circumferentially, while $s$ is measured longitudinally.}
		\label{fig:cylinder}
	\end{center}
\end{figure}

\subsubsection{The shape pencil}
The coefficient matrices of the first and second fundamental forms constitute a positive-definite symmetric pencil of matrices. This characterization comes from the fact that the metric tensor is symmetric positive definite, while the matrix associated with the second fundamental form is symmetric.
In combination they form a {\em positive symmetric matrix pencil}
\be
	{\mathcal S}(\lambda) = \lambda I_P - \Pi_P;
	\label{eq:ShapePencil}
\ee
we call this the {\em shape pencil} (or {\em Weingarten pencil}).
The name `matrix pencil' is a long-standing mathematical metaphor in which a collection of geometric, or algebraic objects is `bundled' together, like the bristles of a brush that converge to the point (of a pencil).
The term was adapted into linear algebra to describe a combination of two matrices, typically written as $L(\lambda) = A + \lambda B$ where $\lambda$ is a scalar parameter \cite{GantmacherV2}. Matrix pencils have eigenvalues and eigenvectors, since $(\lambda I_P - \Pi_P){\bm w}=0$ has solutions for the eigenvalues $\lambda$ and eigenvectors ${\bm w}$. 
Positive definite symmetric pencils have real eigenvalues with their corresponding eigenvectors orthogonal in the metric tensor. See \cite{Parlett1998} Theorem 15.3.3. The eigenvalues are the {\em principal curvatures} and the eigenvectors are the
{\em principal curvature directions}\textemdash the tangent directions along which the surface
curvature is most and least.
One can also analyse the problem ${\bm v}^T(\lambda I_P - \Pi_P)=0$ in which the ${\bm v}$'s are the dual eigenvectors. The eigenvectors and dual eigenvectors are important when characterising ${\mathcal S}(\lambda)({\bm v},{\bm w})$ as a bilinear form. 

\subsubsection{Connection to the directional derivative}
The bilinear form associated with the pencil is %the family of bilinear forms defined by the linear combination:
\be
{\mathcal S}(\lambda) ({\bm v},{\bm w}) = \lambda I_P ({\bm v},{\bm w}) - \Pi_P({\bm v},{\bm w}).
\ee
The pairs $(\lambda, {\bm w})$ such that ${\mathcal S}(\lambda) ({\bm v},{\bm w}) = 0$ for all ${\bm v}$  are the principal curvatures and the principal curvature directions. 
Using the directional-derivative representation \eqref{eq:DirDer},
the condition ${\mathcal S}(\lambda)({\bm v},{\bm w}) = 0$ becomes
\begin{equation}
	-(\nabla_{\bm w} {\mathcal N}_P)\cdot{\bm v} = \lambda\, I_P({\bm v},{\bm w}),
\end{equation}
or equivalently $(\nabla_{\bm w} {\mathcal N}_P + \lambda\,{\bm w})\cdot{\bm v} = 0$; note that the dot represents the metric-induced inner product (${\bm v} \cdot {\bm w} = I_P({\bm v},{\bm w})$). Since $\bm v$ is in the tangent plane, but is otherwise arbitrary, the vector $(\nabla_{\bm w} {\mathcal N}_P + \lambda {\bm w})$ must be either zero, or in the normal direction.
The first option must hold, because $(\nabla_{\bm w} {\mathcal N}_P + \lambda {\bm w})$ is in the tangent plane.
Thus
\be
(\nabla_{\bm w} {\mathcal N}_P) = - \lambda {\bm w}. \label{eq:PrincipalDirEq}
\ee
If ${\bm w}$ is a principal direction with principal curvature $\lambda$, moving in the ${\bm w}$ direction causes the surface normal to tilt at a rate $\lambda$ in the same direction.
A positive $\lambda$ means that the normal tips in the direction of travel.
The two principal directions are orthogonal in the surface metric, so they form a
natural curvature-aligned frame at each point on the surface.

\subsection{Measures of curvature}
The principal curvatures are the eigenvalues of the shape pencil and are denoted $\kappa_1$ and $\kappa_2$,
with corresponding orthogonal principal directions ${\bm w}_1$ and ${\bm w}_2$.
Two scalar combinations of $\kappa_1$ and $\kappa_2$ are important:
\be
	K = \kappa_1\kappa_2 = \frac{\det(\Pi_P)}{\det(I_P)},
	\qquad
	H = \frac{\kappa_1+\kappa_2}{2},
\ee
in which $K$ is the {\em Gaussian curvature} and $H$ is the {\em mean curvature}.
For the saddle in Figure\,\ref{fig:saddle}, the principal curvatures are $\kappa_1 =2$ and $\kappa_2=-2$.
The Gaussian curvature is $K=\kappa_1 \kappa_2 = \frac{\det(\Pi_P)}{\det(I_P)}=-4$; $K>0$ is indicative of a dome-like surface, $K<0$ is indicative of a saddle, and $K=0$ means that the surface is locally flat in at least one direction.
The mean curvature has been famously applied to soap films and bubbles.
A soap film, having no pressure difference across it, must have zero mean curvature everywhere, whence $H=0$.
If the surface curves ``up'' in one direction, it must curve ``down'' by the same amount in the perpendicular direction to ``keep things balanced.''
There is ambiguity associated with both the mean and Gaussian curvatures, since the former says nothing about how much ``up'' and how much ``down'' there is, while the latter cannot distinguish between planes and cylinders.

The Gaussian curvature appears initially to depend on how the surface is embedded in space.
However, Gauss proved that $K$ can be expressed in terms of $E$, $F$, $G$, and their first and second derivatives\textemdash no reference to the embedding in space is required. This is the {\em Theorema Egregium} (the Remarkable Theorem). This means that $K$ is an intrinsic property of the surface\textemdash a bug crawling on the surface knowing only the metric tensor can measure the Gaussian curvature using only distances and angles. So the bug knows whether it is crawling on a sphere, a plane, or a saddle\textemdash without ever look beyond the surface itself.
A bug on a wire has $(dl)^2$ as the metric tensor that can only measure distance. It cannot detect curvature, or the embedding (in ${\mathbb R}^2$, ${\mathbb R}^3$, or some higher-dimensional space). The wire could be a straight line (in ${\mathbb R}^2$ or ${\mathbb R}^3$), a circle, a helix, or a twisted ellipse\textemdash every wire is perceived intrinsically as linear.

\section{Acceleration and Gauss' formula} \label{sec:Gauss}
The acceleration of point $\bm p$ is given by the five-term formula
\be
\ddot{\bm p} = \ddot{s} \partial_s {\bm p} + \ddot{n} \partial_n {\bm p} + \dot{s}^2 \partial_{ss} {\bm p}
+ 2 \dot{s} \dot{n} \partial_{sn} {\bm p} + \dot{n}^2 \partial_{nn} {\bm p}; \label{eq:acc}
\ee
this is not the classical five-term acceleration formula in Newtonian mechanics.
Unlike $\partial_s {\bm p}$ and $\partial_n {\bm p}$, the second-order partial derivatives in \eqref{eq:acc} are not typically in the tangent plane ${\mathcal T}_{\mathcal M}^P$. 
For that reason we project them onto each of the independent vectors $\partial_s {\bm p}$, $\partial_n {\bm p}$ and ${\mathcal N}_P$ to give 
\beaw
\partial_{ss} {\bm p} &= \Gamma_{ss}^s  \partial_s {\bm p} + \Gamma_{ss}^n  \partial_n {\bm p} + L  {\mathcal N}_P \nonumber \\
\partial_{sn} {\bm p} &= \Gamma_{sn}^s  \partial_s {\bm p} + \Gamma_{sn}^n  \partial_n {\bm p} + M  {\mathcal N}_P \nonumber \\
\partial_{ns} {\bm p} &= \Gamma_{ns}^s  \partial_s {\bm p} + \Gamma_{ns}^n  \partial_n {\bm p} + M  {\mathcal N}_P \nonumber \\
\partial_{nn} {\bm p} &= \Gamma_{nn}^s  \partial_s {\bm p} + \Gamma_{nn}^n  \partial_n {\bm p} + N  {\mathcal N}_P. \label{eq:CS}
\eeaw
The {\em Christoffel symbols} give the magnitudes of the tangential components of 
$\ddot{\bm p}$ and are given by $\Gamma_{ij}^k$.

\subsection{Centripetal acceleration}
The second derivative of the position vector decomposes into a tangential part (the Christoffel symbols) and a normal part (the second fundamental form). In Einstein notation, 
\be
\partial_{ij} {\bm p} = \Gamma_{ij}^k \partial_k {\bm p} + h_{ij} {\mathcal N}_P, \label{eq:GA}
\ee
in which $h_{ij}$ represents the $ij^{th}$ term in the second fundamental form. Equations \eqref{eq:CS}, or alternatively \eqref{eq:GA}, are known as Gauss's equations.
Substituting (\ref{eq:GA}) into \eqref{eq:acc}, which is $\ddot{\bm p}=\ddot{u}^i \partial_i {\bm p} + \dot{u}^i \dot{u}^j \partial_{ij} {\bm p}$, gives
\be
\ddot{\bm p}= \left(\ddot{u}^{k} + \Gamma_{ij}^{k} \dot{u}^{i} \dot{u}^{j}\right) \partial_k {\bm p} + \Pi_P(\dot{\bm p},\dot{\bm p}) {\mathcal N}_P, \label{eq:acc_decom}
\ee
with $\Pi_P(\dot{\bm p},\dot{\bm p}) = h_{ij} \dot{u}^i \dot{u}^j $.
The first term lies in the tangent plane, with the second term normal to it. 
Looking at the shape pencil, we see that bending is balanced against speed when $\Pi_P(\dot{\bm p},\dot{\bm p})=\kappa_n I_P(\dot{\bm p},\dot{\bm p})$ for the normal curvature $\kappa_n$ (this is the centripetal acceleration ${ a}_c=\kappa_n v^2$), that is
\be
\Pi_P(\dot{\bm p},\dot{\bm p}) = \kappa_n \|\dot{\bm p}\|^2.
\ee
If $\frac{\dot{\bm p}}{\|\dot{\bm p}\|} = \cos \theta \, {\bm e}_1 + \sin \theta \, {\bm e}_2$, where ${\bm e}_1$ and  ${\bm e}_2$ are the principal vectors of the shape pencil, there holds:
\begin{equation}
\begin{aligned}
\Pi_P \left(\frac{\dot{\bm p}}{\|\dot{\bm p}\|},\frac{\dot{\bm p}}{\|\dot{\bm p}\|} \right) ={}& \cos^2\theta\, \Pi_P({\bm e}_1, {\bm e}_1) \\
&{}+ 2 \sin \theta \cos \theta \,\Pi_P({\bm e}_1,{\bm e}_2)
+ \sin^2 \theta \, \Pi_P({\bm e}_2,{\bm e}_2).
\end{aligned}
\end{equation}
But $\Pi_P({\bm e}_1,{\bm e}_2) = 0$ due to the orthogonality of the eigenvectors of the shape pencil. Since $\Pi_P({\bm e}_i,{\bm e}_i) = \kappa_i I_P ({\bm e}_i,{\bm e}_i)$, 
we obtain
\be
\kappa_n = \kappa_1 \cos^2\theta + \kappa_2 \sin^2\theta \label{eq:Euler}
\ee
which is Euler's formula; $\theta$ is the angle between ${\bm e}_1$ and $\dot{\bm p}$. Thus
\be
\ddot{\bm p}= (\ddot{u}^k + \Gamma_{ij}^k \dot{u}^i \dot{u}^j) \partial_k {\bm p} + \kappa_n  \|\dot{\bm p}\|^2 {\mathcal N}_P.
\ee
The $\kappa_n \|\dot{\bm p}\|^2$ term is the centripetal acceleration due to the surface's curvature. In other words, the second fundamental form provides the curvature-related force normal to ${\mathcal T}_{\mathcal M}^P$. The Christoffel terms act like fictitious forces in the curvilinear coordinate system, the $\Gamma_{ij}^k \dot{u}^i \dot{u}^j$ term mimics centrifugal and Coriolis effects in the tangent plane.

\subsection{An important formula} The Christoffel symbols
can be computed knowing only the first fundamental form and its first-order derivatives. That is \cite{Shifrin2006}:
\be
\TBT{\Gamma_{ss}^s}{\Gamma_{sn}^s}{\Gamma_{nn}^s}{\Gamma_{ss}^n}{\Gamma_{sn}^n}{\Gamma_{nn}^n} = \tbt{E}{F}{F}{G}^{-1}
\TBT{\frac{1}{2} \partial_s E}{\frac{1}{2} \partial_n E}
{\partial_n F - \frac{1}{2} \partial_s G}
{\partial_s F - \frac{1}{2} \partial_n E}{\frac{1}{2} \partial_s G}{\frac{1}{2} \partial_n G}, \label{eq:CS2}
\ee
where $\Gamma_{sn}^s=\Gamma_{ns}^s$ and $\Gamma_{sn}^n=\Gamma_{ns}^n$ are assumed.

\subsection{Geodesics}
We consider next the force component parallel to the plane ${\mathcal T}_{\mathcal M}^P$. 
In the absence of dissipation, the kinetic energy of a point mass moving unimpeded in the plane is constant and given by ${\mathcal E}=\frac{m}{2} \dot{\bm p} \cdot \dot{\bm p}$. Thus
\beaw
d_t {\mathcal E} &= 0 = m \dot{\bm p} \cdot \ddot{\bm p} \nonumber \\
&= m \obt{\dot s}{\dot n} \tbo{\partial_s {\bm p} }{\partial_n {\bm p}} \cdot \ddot{\bm p}\nonumber \\
&= m \obt{\dot s}{\dot n} \tbt{\partial_s {\bm p}\cdot\partial_s {\bm p}}{\partial_s {\bm p}\cdot\partial_n {\bm p}}{\partial_n {\bm p}\cdot\partial_s {\bm p}}{\partial_n {\bm p}\cdot\partial_n {\bm p}} \left(\tbo{\ddot{s}}{\ddot{n}}
\!+\! \dot{s}^2\tbo{\Gamma^s_{ss}}{\Gamma^n_{ss}} \!+\! 2 \dot{s} \dot{n} \tbo{\Gamma^s_{sn}}{\Gamma^n_{sn}} \!+\! \dot{n}^2\tbo{\Gamma^s_{nn}}{\Gamma^n_{nn}} \right)\nonumber \\
&= m \obt{\dot s}{\dot n} \tbt{E}{F}{F}{G} \left(\tbo{\ddot{s}}{\ddot{n}} +\tbo{\Xi_s} {\Xi_n} \right),
\eeaw
or what is the same
\be
\tbo{\ddot{s}}{\ddot{n}} +\tbo{\Xi_s}{\Xi_n} = 0. \label{eq:EOM}
\ee
Both $\Xi_s$ and $\Xi_n$ are quadratic forms in the {\em Christoffel symbols}:
\be
\Xi_k = \obt{\dot s}{\dot n} \tbt{\Gamma_{ss}^k}{\Gamma_{sn}^k}{\Gamma_{sn}^k}{\Gamma_{nn}^k} \tbo{\dot s}{\dot n} = \Gamma^k_{ij} \dot{u}^i \dot{u}^j. \label{eq:GeoForms}
%
 %\mbox{~and~} \Xi_n = 
%\obt{\dot s}{\dot n} \tbt{\Gamma_{ss}^n}{\Gamma_{sn}^n}{\Gamma_{sn}^n}{\Gamma_{nn}^n} \tbo{\dot s}
%{\dot n} = \Gamma^n_{ij} \dot{u}^i \dot{u}^i.
\ee
Equation \eqref{eq:EOM} follows by direct calculation, and  
requires the chain rule, \eqref{eq:acc}, \eqref{eq:CS},  and the definition of the entries in the {\em first} fundamental form.
Equation \eqref{eq:EOM} generates {\em geodesic trajectories} on $\mathcal M$\textemdash geodesics generalise the notion of rectilinear motion to curved surfaces.
Geodesics are sometimes described, incorrectly, as being `shortest paths' between points on $\mathcal M$. 
In the case of a sphere, the geodesics are great circles with their centres the centre of the sphere.
Imagine going between points $A$ and $B$ on a sphere; going the `long way around' is still a geodesic path, but it is not a shortest path. In general, the geodesic equations are nonlinear and typically difficult to solve.

\section{The elliptic cone} \label{sec:Elliptic}
In the remainder of this section we aim to pull these ideas together using the elliptic cone that has been used as a proxy for a NASCAR track in optimal control studies \cite{Limebeer2025B}. A road surface on the interior of an elliptic cone can be described by  
\[
\frac{s^2}{a^2} + \frac{n^2}{b^2} = \frac{z^2}{c^2}
\]
in Cartesian coordinates, with gravity pointing downwards. A position vector on the interior of the cone can be described in Monge coordinates by
\be
{\bm p} = \left[ s,\,n,\,f(s,n) \right],
\ee
in which $f(s,n) = c\sqrt{\frac{s^2}{a^2} + \frac{n^2}{b^2}}$.

The metric tensor is given by
\be
I_P = \tbt{E}{F}{F}{G},
\ee
where
\beaw
E &= 1 + (\partial_s f)^2 = 1 +\frac{c^2s^2}{a^4 R^2} \\
F &= \partial_{s} f \partial_{n} f = \frac{c^2sn}{a^2 b^2 R^2} \\
G &= 1 + (\partial_n f)^2 = 1 +\frac{c^2n^2}{b^4 R^2},
\eeaw
with $R=\sqrt{\frac{s^2}{a^2} + \frac{n^2}{b^2}}$. The metric is singular at the origin $(0,0)$; if one looks at the denominators, $R$ goes to zero at the vertex of the cone. This makes sense physically because the curvature at the tip of a cone is undefined\textemdash it's a point where the surface is no longer smooth.

The surface normal at $\bm p$ is given by
\beaw
{\mathcal N}_P &= \frac{\partial_s {\bm p} \times \partial_n {\bm p}}{\| \partial_s {\bm p} \times \partial_n {\bm p} \|} \\
&= \frac{[-\partial_s f,\,-\partial_n f,1]}{\Delta},
\eeaw
where $\Delta = \sqrt{1+\frac{c^2}{R^2}\left(\frac{s^2}{a^4}+\frac{n^2}{b^4}\right)}$.
Since
$\partial_{ss} f = \frac{cn^2}{a^2 b^2 R^3}$, $\partial_{sn} f = -\frac{csn}{a^2 b^2 R^3}$, and 
$\partial_{nn} f = \frac{cs^2}{a^2 b^2 R^3}$, the matrix in the second fundamental form is
\be
\Pi_P = \frac{c}{a^2 b^2 R^3 \Delta} \tbt{n^2}{-sn}{-sn}{s^2}.
\ee
The shape pencil is given by
\be
{\mathcal S} = (\lambda I_P - \Pi_P),
\ee
which has principal curvatures (eigenvalues)
\be
\kappa_1 = \frac{c(s^2+n^2+c^2R^2)}{a^2b^2R^3\Delta^3}
\ee
and $\kappa_2 =0$ (since $\mathcal S$ is singular).
%with $\Delta = 1+\frac{c^2}{R^2}\left(\frac{s^2}{a^4}+\frac{n^2}{b^4}\right)$.

In order to find the geodesic curve between points $(s_A,n_A)$ and $(s_B,n_B)$, the initial-value problem \eqref{eq:EOM}, which has the four initial conditions $[s(0),\,n(0),\,\dot{s}(0),\,\dot{n}(0)]$, must be recast as a two-point boundary-value problem (TPBVP) with the mixed boundary
conditions $[s_A(0),\,n_A(0),\,s_B(T),\,n_B(T)]$; $T$ is the unknown terminal time.
In the TPBVP there are
two coupled nonlinear differential equations with four boundary conditions, which are all position-related.
The Christoffel symbols are computed in $(s,n)$ coordinates using \eqref{eq:CS2}. There are two problems with this approach: (1) even for this relatively simple problem the Christoffel symbols are cumbersome, and (2) the TPBVP must be solved by a method such as multiple-shooting, which can be numerically poorly conditioned. Also, while mathematically correct, the TPBVP route lacks an important geometric insight. It turns out that there is a better way$\ldots$

Since $\kappa_2 =0$, the surface is locally flat (in one direction). The geometric significance of this 
is that the surface is {\em developable} (it can be `flattened out' without tearing), which allows one to construct straight-line geodesics in a transformed coordinate system \cite{Shifrin2006}.

In a first step we parametrise the cone as 
\[
{\bm p}(s,n) = \underbrace{[a \cos s, b \sin s,c]}_{\alpha (s)} + n \underbrace{[-a \cos s, -b \sin s,-c]}_{\beta (s)},
\]
which is a {\em ruled surface} with rulings $\beta(u)$ and directrix $\alpha (u)$. This step allows one to draw a 2D coordinate grid onto the 3D cone, much like drawing lines of latitude and longitude onto a 3D globe. The angular position around the cone is given by $s \in [0\,2 \pi]$, while $n$ represents the distance along the straight-line rulings from the rim down to the vertex (at $n=1$). When a cone is unrolled flat onto a plane, the vertex
maps to the origin of the plane, and the straight rulings map to straight rays radiating from the origin. The infinitesimal displacement on the surface is:
\[
dl^2 = d{\bm p} \cdot d{\bm p} = dn^2 (\partial_n {\bm p} \cdot \partial_n {\bm p}) + 2 dsdn (\partial_s {\bm p} \cdot \partial_n {\bm p})
+ ds^2 (\partial_s {\bm p} \cdot \partial_s {\bm p}).
\]
In the ruling direction $(dn)$, the coefficient is:
\[
\partial_n {\bm p} \cdot \partial_n {\bm p} = a^{2}\cos ^{2}s+b^{2}\sin ^{2}s+c^{2}.
\]

In the flat plane, the distance from the origin is measured by the polar radius $r$. For the mapping to be an isometry along these rulings,
$r$ must match the true geometric distance from the vertex along the 3D ruling generator:
\be
r(s,n)=n \partial_n {\bm p}(s,n) \cdot \partial_n {\bm p}(s,n) = n \sqrt{a^{2}\cos ^{2}s+b^{2}\sin ^{2}s+c^{2}}.
\ee
Once ${\bm p}(s,n)$ is expressed in flat polar coordinates $(r, \phi)$, it can be transformed into 2D Cartesian coordinates $(U, V)$
using $(U=r\cos \phi ,\, V=r\sin \phi)$, which ensures that the flat Pythagorean metric 
\[
dU^{2}+ dV^{2}= dr^{2}+r^{2} d\phi ^{2}
\]
holds.

To find the relationship for the angle $\phi(s)$, we equate the flat-plane area element $r dr d \phi$ to the cone's surface-area element $\| \partial_s {\bm p} \times \partial_n {\bm p} \|$, knowing that the mapping between the two is isometric, to obtain
\be
\phi(s) = \int_0^s \frac{\sqrt{a^2 b^2 + c^2(a^2 \sin^2 \tau +b^2 \cos^2 \tau)}}{a^2 \cos^2 \tau + b^2 \sin^2 \tau + c^2} d \tau,
\ee
which is an elliptic integral from which geodesics can be easily found; filling these details is left as an exercise for the reader.

The geodesics are straight lines in the $(U,V)$ plane, which can be mapped back onto the cone's surface using the inverse mapping
\be
n = \frac{r}{\sqrt{a^2 \cos^2 s + b^2 \sin^2 s + c^2}}
\ee
and by inverting $\phi(s)$. 
The developed surface is an ellipsoidal sector in the $(r, \phi)$ plane. The total elapsed angle is
\be
\Phi_{total} = \int_0^{2 \pi} \frac{\sqrt{a^2 b^2 + c^2(a^2 \sin^2 \tau +b^2 \cos^2 \tau)}}{a^2 \cos^2 \tau + b^2 \sin^2 \tau + c^2} d \tau
\ee
which is less than $2 \pi$ (the deficit angle) because the cone is not flat at the vertex.
The elliptic integral tells us that the cone's intrinsic geometry is flat, but non-uniformly distorted.
Figure\,\ref{fig:GeoD} shows a geodesic curve on the surface of an elliptic cone that was computed using the code in Appendix\,\ref{App:Geo}.

\begin{figure*}[h!] 
\begin{center}
\includegraphics[trim=0mm 80mm 0mm 85mm, clip, width=0.78\textwidth]{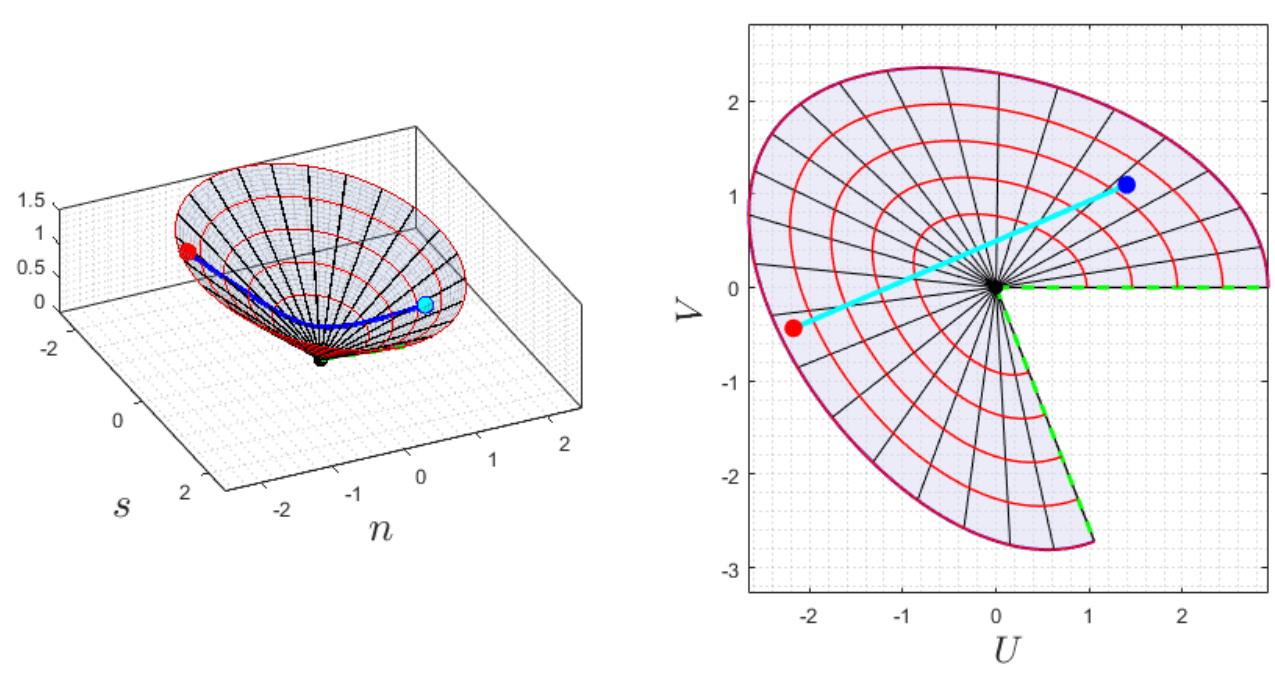}
\caption{A geodesic curve between two points on the surface of an elliptic cone ($a\,=\,1.5$, $b\,=\,1.0$ and $c\,=\,0.9$), and with $(s_A,\,n_A) = (1,\,1)$ and $(s_B,\,n_B) = (-1,\,-1.4)$. The right-hand figure shows the straight line geodesic on the developed cone. The geodesic arclength is 3.88, with the deficit angle 68.8$^o$.}
\label{fig:GeoD}
\end{center}
\end{figure*}

This example isn't just about studying cones. It's about using geometry to avoid analysis and the solution of nonlinear differential equations when a surface is developable. The difficult part is finding the flattening map (solving an elliptic integral in this case). After that, the geodesics are straight lines that essentially come for free.

\section{Dynamics}\label{sec:Dynamics}
There is nothing interesting about particle motions on stationary planar surfaces, which is not the case when the underlying surface is curved.
Here we consider the dynamics of a point mass $m$ with initial conditions $s_0$, $n_0$, $\dot{s}_0$, and $\dot{n}_0$. Since the system is conservative, the particle's total energy is
conserved (with a constant Hamiltonian\textemdash $d_t {\mathcal H}=0$), with its Lagrangian given by
\be
{\mathcal L}(s,n,\dot{s},\dot{n})=\frac{m}{2} \| \dot{\bm p} \|^2 - mgcR.
\ee 
The associated Euler-Lagrange equations are
\be
d_t \partial_{\dot{s}} {\mathcal L} = \partial_s {\mathcal L} \mbox{~~~~and~~~~} d_t \partial_{\dot{n}} {\mathcal L} = \partial_n {\mathcal L}
\ee
These can be solved to yield
\beaw
d_t(E\dot{s}+F\dot{n}) &= \frac{1}{2} \left( \partial_s E\dot{s}^2 +2 \partial_s F \dot{s}\dot{n} +\partial_s G \dot{n}^2 \right) -\frac{g cs}{a^2 R} \\
d_t(F\dot{s}+G\dot{n}) &= \frac{1}{2} \left( \partial_n E\dot{s}^2 +2 \partial_n F \dot{s}\dot{n} 
+\partial_n G \dot{n}^2 \right) -\frac{g c n}{b^2 R},
\eeaw
which can be arranged to read
\beaw
\tbo{\ddot{s}}{\ddot{n}} &= \tbt{E}{F}{F}{G}^{-1}
\left(\TBT{\frac{1}{2} \partial_s E}{\frac{1}{2} \partial_n E}
{\partial_n F - \frac{1}{2} \partial_s G}
{\partial_s F - \frac{1}{2} \partial_n E}{\frac{1}{2} \partial_s G}{\frac{1}{2} \partial_n G} 
\Tbo{\dot{s}^2}{\dot{s}\dot{n}}{\dot{n}^2}
- \frac{gc}{R} \tbo{\frac{s}{a^2}}{\frac{n}{b^2}} \right)  \nonumber \\
&= -\left( \dot{s}^2\tbo{\Gamma^s_{ss}}{\Gamma^n_{ss}} 
+ 2 \dot{s} \dot{n} \tbo{\Gamma^s_{sn}}{\Gamma^n_{sn}}
+ \dot{n}^2\tbo{\Gamma^s_{nn}}{\Gamma^n_{nn}} \right)  
- \frac{gc}{R} \tbt{E}{F}{F}{G}^{-1}\tbo{\frac{s}{a^2}}{\frac{n}{b^2}} \nonumber \\
&= -\tbo{\Xi_s}{\Xi_n} - \frac{gc}{R} \tbt{E}{F}{F}{G}^{-1}\tbo{\frac{s}{a^2}}{\frac{n}{b^2}}
\eeaw
using the chain rule, \eqref{eq:EOM}, and \eqref{eq:GeoForms}. The first term corresponds to the geodesic equations \eqref{eq:EOM}, while the second term is the conservative force due to gravity.
In Einstein notation,
\be
\ddot{u}^{k}=-\Gamma_{ij}^{k}\dot{u}^{i}\dot{u}^{j}-gc\,g^{k \ell}\partial_{\ell}R, \label{eq:index_form}
\ee
where $g^{k \ell}$ indicates the inverse of the metric tensor $g_{k \ell}$.

Equation \eqref{eq:index_form} can be obtained directly from Newton's second law. To do this,
we need simply to add the gravity-related force action on the point mass to the curvature-related forces in the geodesic equations; we can assume $m=1$ without loss of generality.
To begin, we decompose gravity into the normal and tangent plane directions to obtain 
\be
{\bm g} = \alpha \partial_s {\bm p} + \beta \partial_n {\bm p} +  \gamma {\mathcal N}_P.
\label{eq:GravDecp}
\ee
From which we obtain
\beaw
\gamma &= ({\bm g} \cdot {\mathcal N}_P) \nonumber \\
\tbt{\partial_s {\bm p} \cdot \partial_s {\bm p}}{\partial_s {\bm p} \cdot \partial_n {\bm p}}
{\partial_n {\bm p} \cdot \partial_s {\bm p}}{\partial_n {\bm p} \cdot \partial_n {\bm p}} \tbo{\alpha}{\beta} &=
\tbo{\partial_s {\bm p} \cdot {\bm g} }{\partial_n {\bm p} \cdot {\bm g}},
\eeaw
and so
\be
\tbo{\alpha}{\beta} = \tbt{E}{F}{F}{G}^{-1} \tbo{\partial_s {\bm p} \cdot {\bm g} }{\partial_n {\bm p} \cdot {\bm g}}
= g_c g^{k \ell} \partial_\ell R. \label{eq:GravVec}
\ee
In this context the normal reaction force is dynamically non-contributory, with the $g_c g^{k \ell} \partial_\ell R$ terms added to the geodesic equation \eqref{eq:EOM} to obtain equations \eqref{eq:index_form}. 

\begin{remark}
Particles subject to external force systems that have no tangential components follow geodesic trajectories described by the Christoffel symbols.
Normal forces produce non-working reactions, leaving the intrinsic trajectory geodesic.
All other external forces such as aerodynamic drag, driving and braking forces, and tangential gravity components cause the path to deviate from geodesic curves.
This composition of ideas is why differential geometry is such a powerful tool for constraint mechanics.
\end{remark}

In Figure\,\ref{fig:ConeD} a unit-mass particle is released from an initial position on an elliptic cone; these curves were computed using the simulation code in Appendix\,\ref{App:DynEllip}.  
Since this system is conservative, the particle follows a constant-energy trajectory, which involves continuous energy trading between kinetic and potential energy, while keeping their sum constant (at the initial release potential energy). The resulting motion is only periodic for special initial conditions.
Another view of the particle's motion can be seen in the phase portraits shown in Figure\,\ref{fig:PPort},
which are useful when trying to distinguish between periodic, quasi-periodic, and chaotic behaviours.

\begin{figure*}[ht] 
	\begin{center}
		\includegraphics[trim=10mm 85mm 10mm 90mm, clip, width=0.7\textwidth]{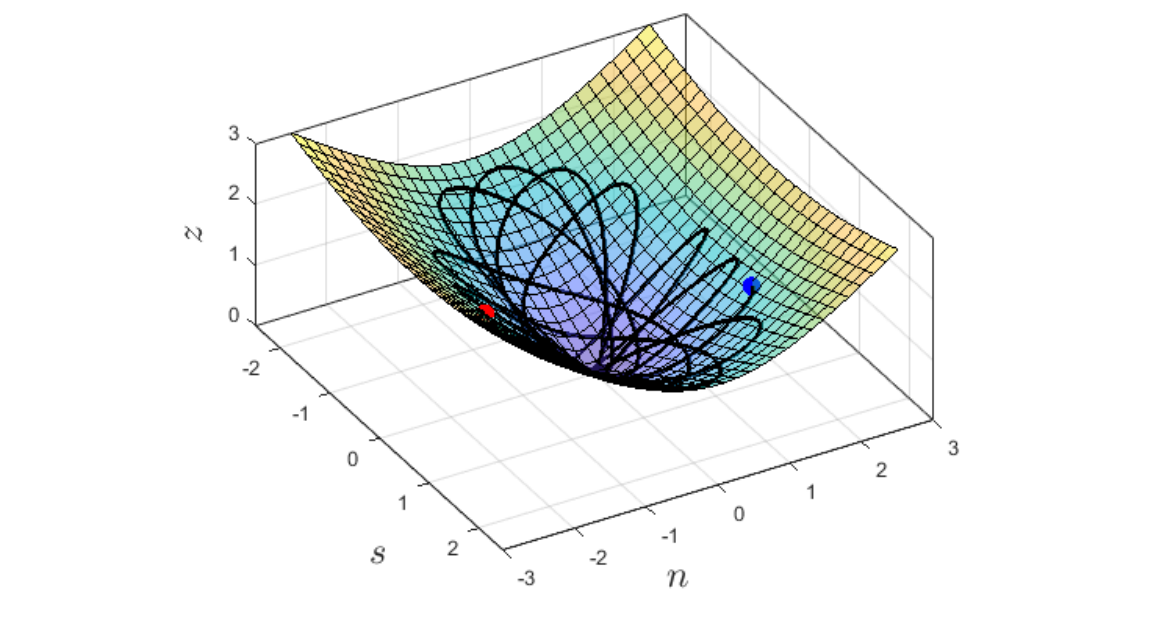}
		\caption{Motion of a unit mass on the interior surface of an elliptic cone. The blue dot marks the start of the trajectory, while the red dot marks its end.}
		\label{fig:ConeD}
	\end{center}
\end{figure*}

\begin{figure*}[ht] 
	\begin{center}
		\includegraphics[trim=10mm 85mm 10mm 90mm, clip, width=0.7\textwidth]{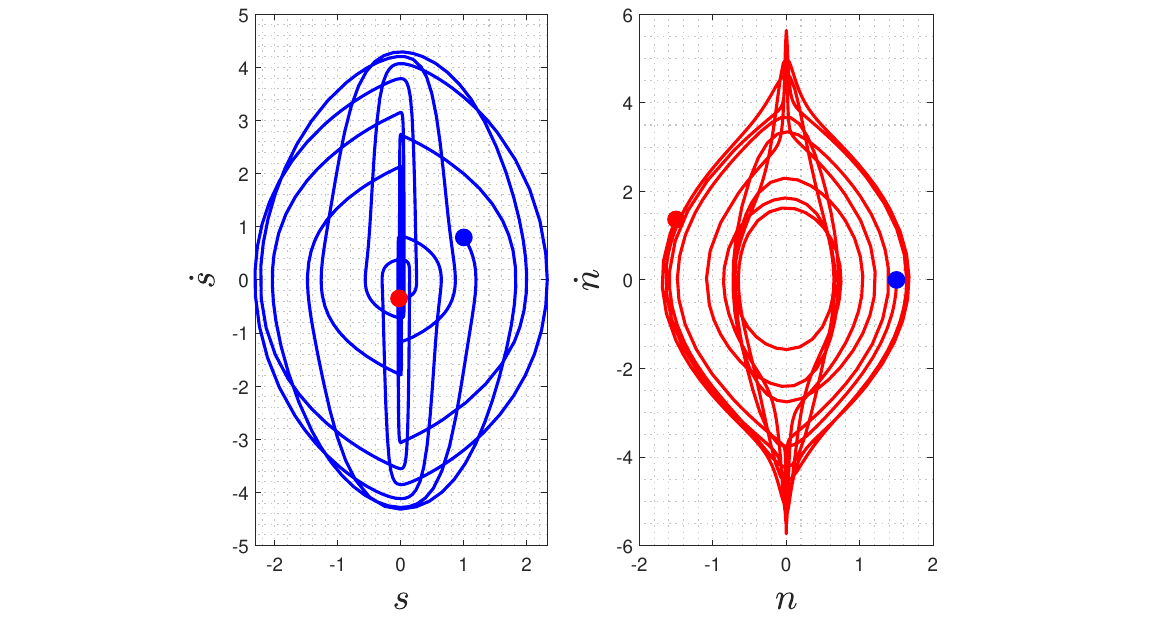}
		\caption{Phase portrait of the conservative trajectory of a unit mass on the interior surface of an elliptic cone.}
		\label{fig:PPort}
	\end{center}
\end{figure*}

\section{Connection to vehicle dynamics}\label{sec:VehicleDynamics}
In this section the geometric machinery described above is recast into a form usable in vehicle dynamics problems. 
The central idea is that the influence of the road surface is fully characterised by three geometric objects: the metric tensor $I_P$, a Jacobian $J$, and the shape operator $\mathcal S$. The metric tensor is also used to express
gravity in the vehicle's body-fixed coordinate system using arguments that resemble \eqref{eq:GravDecp}.
Collectively, these quantities provide a complete interface between the road geometry and the vehicle model \cite{Limebeer2025}.
Examples of the use of models of this kind can be found in \cite{VandeMerwe2026, Limebeer2026}.

To begin, we consider two alternative expressions for the velocity of the mass centre of a vehicle on a curved surface.
\[
{\bm v} = {\bm e}_x^b u + {\bm e}_y^b v \mbox{~~~and~~~} {\bm v} = \partial_s {\bm p} \dot{s} + \partial_n {\bm p} \dot{n}
\]
in which ${\bm e}_x^b$ and ${\bm e}_y^b$ represent an orthonormal body-fixed coordinate system. From these we obtain 
\beaw
{\bm e}_x^b \cdot {\bm v} = u = {\bm e}_x^b \cdot \partial_s {\bm p}\, \dot{s} + {\bm e}_x^b \cdot  \partial_n {\bm p} \,\dot{n} \nonumber \\
{\bm e}_y^b \cdot {\bm v} = v = {\bm e}_y^b \cdot \partial_s {\bm p}\, \dot{s} + {\bm e}_y^b \cdot  \partial_n {\bm p}\, \dot{n}.
\eeaw
These can be re-packaged as
\be
\tbo{u}{v}= \overbrace{\tbt{{\bm e}_x^b \cdot \partial_s {\bm p}}{{\bm e}_x^b \cdot  \partial_n {\bm p}}{{\bm e}_y^b \cdot \partial_s {\bm p}}
{{\bm e}_y^b \cdot  \partial_n {\bm p}}}^{J} \tbo{\dot{s}}{\dot{n}}  \Rightarrow \tbo{\dot{s}}{\dot{n}} = J^{-1} \tbo{u}{v}, \label{eq:LinEqnA}
\ee
in which the Jacobian $J$ represents a transformation from the curvilinear coordinate system $(\partial_s {\bm p},\,\partial_n {\bm p})$ to the body-fixed coordinate system $({\bm e}_x^b,\,{\bm e}_y^b)$. Equation \eqref{eq:LinEqnA} shows how the engine-driven body-fixed velocities $u$ and $v$ translate into the ground-plane velocity components $\dot s$ and $\dot n$.
%\subsection{Coordinate transforms and the metric tensor}

The following lemma formally links the Jacobian to the metric tensor.
\begin{lemma}\label{lem:JacobianMetric}
Let $\mathcal{M}$ be a two-dimensional manifold embedded in $\mathbb{R}^3$ and parameterised by $(s, n)$;
the local tangent space is described by $\mathcal{T}_{\mathcal{M}}^P$.
Let the coordinate basis vectors of the tangent space be given by $\partial_s \bm{p}$ and $\partial_n \bm{p}$.
If $\bm{e}_x^b$ and $\bm{e}_y^b$ are any two orthonormal basis vectors spanning $\mathcal{T}_{\mathcal{M}}^P$, and the coordinate transformation Jacobian matrix $J \in \mathbb{R}^{2 \times 2}$ is defined in \eqref{eq:LinEqnA},
then $I_P = J^T J$.
\end{lemma}

\begin{proof}
Let $\bm{e}_x^b$ and $\bm{e}_y^b$ form an orthonormal basis for the tangent space $\mathcal{T}_{\mathcal{M}}^P$.
The completeness relation 
\be
\bm{a} \cdot \bm{b} = (\bm{a} \cdot \bm{e}_x^b)(\bm{e}_x^b \cdot \bm{b}) + (\bm{a} \cdot \bm{e}_y^b)(\bm{e}_y^b \cdot \bm{b}) \label{eq:completeness}
\ee
applies to any vectors $\bm{a}, \bm{b} \in \mathcal{T}_{\mathcal{M}}^P$.
Expanding $J^T J$ element-by-element reveals that each metric component is given by $g_{ij} = \partial_i \bm{p} \cdot \partial_j \bm{p}$ for $i,j \in \{s,n\}$, which is the first fundamental form.
\end{proof}

\subsection{Singular points and Tikhonov regularization}
Points where $J^{-1}$ does not exist are called singular.
These points can arise when the centre line's centre of curvature lies within the track boundaries \cite{fork2023euclidean};
see Figure\,\ref{fig:Singularity}. To see this, observe that for points on the centre line \cite{Limebeer2025}
\be
\partial_s {\bm p}(s,n) = [\cos \psi(s)(1-n \psi'(s));
\sin \psi(s)(1-n \psi'(s)); \partial_s z(s,n)]
\ee
vanishes when $\partial_s z(s,n)=0$ (the track is horizontal  in the $s$-direction), and $1-n \psi'(s)=0$ ({\em i.e.}
when the vehicle's lateral position $n$ coincides with the local centre of curvature).
At such points, the mapping from body-fixed velocities ${\bm q} = (u,\,v)$ to the surface velocities ${\bm \nu} = (\dot{s},\,\dot{n})$
in \eqref{eq:LinEqnA} loses rank, making $J$ singular.

\begin{figure}[ht]
	\begin{center}
		\includegraphics[trim=5cm 1cm 0mm 0cm, clip, width=0.8\textwidth]{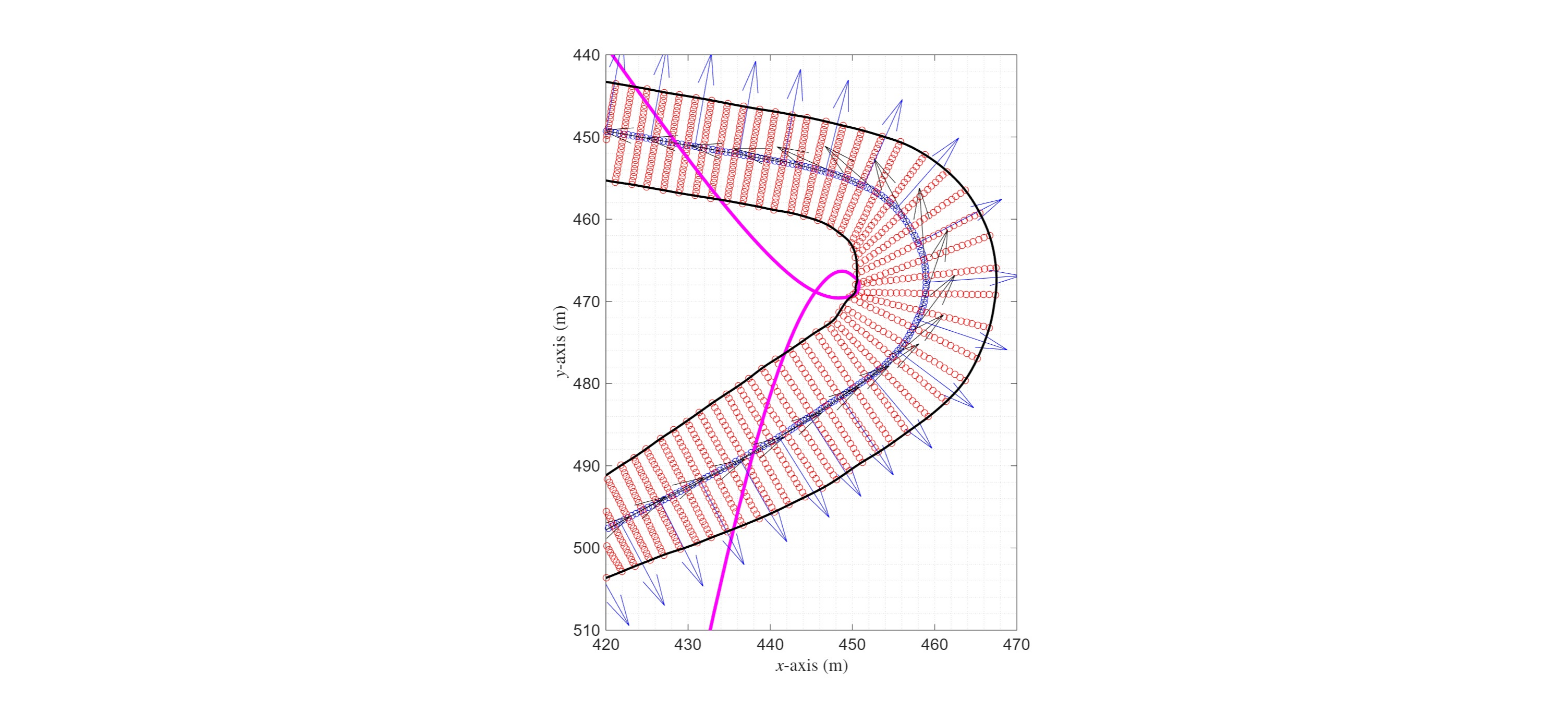}
		\put(-185,102){$S$}
	\end{center}
	\caption{Sharp corner on the Singapore Formula One circuit. The locus of the centre line's centre of curvature is shown as the magenta curve, with a possible singular point of the metric tensor in the neighborhood of point $S$. The red circles are measured LiDAR data points. The black and red arrows are respectively the tangent and principal normal directions. The car's mass centre must avoid these singular points in optimal minimum-lap-time calculations.}
	\label{fig:Singularity}
\end{figure}

This problem can be avoided by re-parametrization, or by using a Tikhonov regularization \cite{tikhonov1977solutions}.
In this (Tikhonov) framework, the inversion of $J$ is replaced by a damped least-squares problem:
\begin{equation}
   E_T({\bm \nu}) = \min_{\bm \nu} \left( \|J {\bm \nu} - {\bm q}\|^2 + \lambda^2 \|{\bf \nu}\|^2 \right), \label{eq:ET}
\end{equation}
where $\lambda > 0$ is the regularizing damping factor. 
Minimizing $E_T({\bm q})$ balances the solution error against a soft penalty on the input energy --- a structure reminiscent of standard LQ control.
From \eqref{eq:ET} and Lemma~\ref{lem:JacobianMetric} we obtain the necessary condition for optimality
\be
  \partial_{\bm \nu} E_T({\bm \nu}) = J^T (J {\bm \nu} - {\bm q}) + \lambda^2 {\bm \nu} = 0 \implies (I_P + \lambda^2 I){\bm \nu} = J^T {\bm q}.
\ee
In the $\lim_{\lambda \rightarrow 0^+}$, this becomes the Moore-Penrose inverse.
Since $I_P + \lambda^2 I$ is positive definite, it can be factorised into Cholesky upper-triangular form $I_P + \lambda^2 I = L^T L$ in which
\be
L_{11} = \sqrt{E + \lambda^2};~L_{12} = \frac{F}{\sqrt{E + \lambda^2}};~ L_{22} = \sqrt{G + \lambda^2 - L_{12}^2},
\ee
using a Gram-Schmidt orthogonalisation process \cite{Golub1996}. This gives
\begin{equation}
\tbo{\partial_s \bm p}{\partial_n \bm p}
=
\tbt{\sqrt{E+\lambda^2}}{0}
	{\dfrac{F}{\sqrt{E+\lambda^2}} } {\sqrt{G+\lambda^2 - L_{12}^2}}
\tbo{\bm q_1}{\bm q_2},
\label{eq:gs_cholesky_basis}
\end{equation}
in which $\bm q_1$ and $\bm q_2$ are the orthonormal Gram--Schmidt vectors. That is
\begin{equation}
	\partial_s \bm p = L_{11}\bm q_1
	\qquad \text{and} \qquad
	\partial_n \bm p = L_{12}\bm q_1 + L_{22}\bm q_2 .
	\label{eq:gs_components}
\end{equation}
Since $\bm q_1$ and $\partial_s \bm p$ are co-linear, the rotation matrix
\[
Q(\chi) = \tbt{\cos\chi}{\sin \chi}{-\sin \chi}{\cos \chi}
\]
allows one to express $u$ and $v$ in the $(\bm q_1, \bm q_2)$ coordinate system, in which 
$\chi$ is the angle between $\partial_s {\bm p}$ and ${\bm e}_x^b$; see \eqref{eq:chiAng} (below) and Figure\,\ref{fig:Angles}.
The matrix $L$ then allows $\dot s$ and $\dot n$ to be evaluates in $(\partial_s {\bm p},\,\partial_n {\bm p})$ coordinates.

The regularised version of Lemma~\ref{lem:JacobianMetric} becomes $I_P + \lambda^2 I = J_{\mbox{reg}}^T J_{\mbox{reg}}$. Then $J_{\mbox{reg}}^T J_{\mbox{reg}} = L^T Q^T Q L$, since $Q$ is orthonormal.
This gives $J_{\mbox{reg}}=QL$ with $J_{\mbox{reg}}$ nonsingular, allowing \eqref{eq:LinEqnA} to be solved.
Equations \eqref{eq:AngVel} and \eqref{eq:GravF} (below) can be dealt with in the same way.

\begin{figure}[ht] 
	\begin{center}
		\includegraphics[trim=30mm 70mm 10mm 0mm, clip, width=0.4\textwidth]{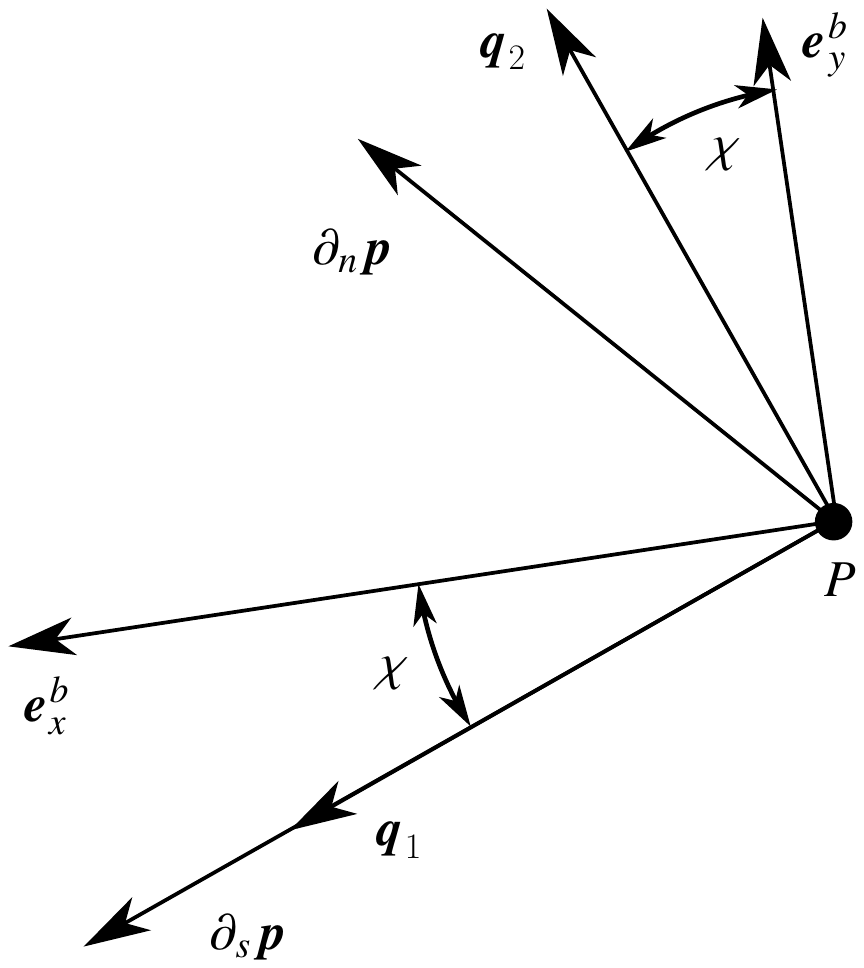}
		\caption{Relative positioning of the tangent plane basis vectors, the vehicle's body-fixed axes, and the Gram-Schmidt vectors ${\bm q}_1$ and ${\bm q}_2$.
			Note that $\partial_s {\bm p}$ and $\partial_n {\bm p}$ are not generally orthogonal, while ${\bm q}_1$ and ${\bm q}_2$ are.}
		\label{fig:Angles}
	\end{center}
\end{figure}

\subsection{Road-induced angular motion and gravity projection}

The second linkage describes the way in which the vehicle's body-fixed angular velocity components are dictated by the road surface (in models without suspension systems)\textemdash the vehicle's pitch and roll dynamics are dictated directly by the road surface.
\begin{lemma}
Let $\Pi_P$ denote the matrix representation of the second fundamental form in the coordinate basis $\{\partial_s \bm{p}, \partial_n \bm{p}\}$, and let $J$ be the coordinate transformation \eqref{eq:LinEqnA}. The vehicle's body-fixed roll and pitch angular velocities, $\omega_y^b$ and $\omega_x^b$, are enforced by the surface parameters and tracking rates as follows:
\be
\tbo{-\omega_y^b}{\omega_x^b} = J^{-T} \Pi_P \tbo{\dot{s}}{\dot{n}}. \label{eq:AngVel}
\ee
\end{lemma}
\begin{proof}
We begin with the rigid-body kinematic definition of the convective time derivative of the surface unit normal vector field $\mathcal{N}_P$ expressed in the vehicle's body-fixed frame:
\be
\dot{\mathcal{N}}_P = \bm{\omega}^b \times \mathcal{N}_P = (\omega_x^b \bm{e}_x^b + \omega_y^b \bm{e}_y^b + \omega_z^b \mathcal{N}_P) \times \mathcal{N}_P = \omega_y^b \bm{e}_x^b - \omega_x^b \bm{e}_y^b. \label{eq:ndot_body}
\ee
Alternatively, applying the chain rule along the surface coordinates, the normal vector rate of change tracks the extrinsic geometry via:
\be
\dot{\mathcal{N}}_P = \dot{s} \partial_s \mathcal{N}_P + \dot{n} \partial_n \mathcal{N}_P.
\ee
Taking the inner product of this expression with the coordinate basis tangent vectors $\partial_i \bm{p}$ for $i \in \{s,n\}$, and invoking the definition of the second fundamental form entries $\Pi_{ij} = -\partial_j \mathcal{N}_P \cdot \partial_i \bm{p}$, we obtain the matrix projection:
\be
\tbo{\dot{\mathcal{N}}_P \cdot \partial_s \bm{p}}{\dot{\mathcal{N}}_P \cdot \partial_n \bm{p}} = -\Pi_P \tbo{\dot{s}}{\dot{n}}; \label{eq:ndot_projection}
\ee
using the identities in \eqref{eq:CS}.
We now project the body-frame definition of $\dot{\mathcal{N}}_P$ from \eqref{eq:ndot_body} onto the same tangent basis vectors:
\be
\tbo{\dot{\mathcal{N}}_P \cdot \partial_s \bm{p}}{\dot{\mathcal{N}}_P \cdot \partial_n \bm{p}} = \tbo{(\omega_y^b \bm{e}_x^b - \omega_x^b \bm{e}_y^b) \cdot \partial_s \bm{p}}{(\omega_y^b \bm{e}_x^b - \omega_x^b \bm{e}_y^b) \cdot \partial_n \bm{p}} = -J^T \tbo{-\omega_y^b}{\omega_x^b}. \label{eq:ndot_jacobian_match}
\ee
Equating \eqref{eq:ndot_projection} and \eqref{eq:ndot_jacobian_match} give:
\be
\tbo{-\omega_y^b}{\omega_x^b} = J^{-T} \Pi_P \tbo{\dot{s}}{\dot{n}}.
\ee
\end{proof}
Where expressions for the pitch and roll angular accelerations are required, we will assume road surfaces that can accommodate
\be
\tbo{-\dot{\omega}_y^b}{\dot{\omega}_x^b} \approx J^{-T} \Pi_P J^{-1} \tbo{\dot{u}}{\dot{v}} \label{eq:AngAcc}
\ee
using \eqref{eq:LinEqnA}. This approximation neglects the convective derivatives of the metric and shape operator, and is valid for surfaces with slowly varying curvature.

The projection of the gravity vector onto body-fixed axes gives:
\be
{\bm g} = G_x {\bm e}_x^b + G_y {\bm e}_y^b + G_z {\bm e}_z^b,
\label{eq:GravDecp2}
\ee
and dotting this equation with $\partial_s{\bm p}$, $\partial_n{\bm p}$, and ${\mathcal N}_P$, yields
\[
\Tbo{\partial_s{\bm p} \cdot {\bm g}}{\partial_n{\bm p} \cdot {\bm g}}{{\mathcal N}_P \cdot {\bm g}}
= \TbT{\partial_s{\bm p} \cdot {\bm e}_x^b}{\partial_s{\bm p} \cdot {\bm e}_y^b}{0}
{\partial_n{\bm p} \cdot {\bm e}_x^b}{\partial_n{\bm p} \cdot {\bm e}_y^b}{0}
{0}{0}{1} \Tbo{G_x}{G_y}{G_z},
\]
from which
\be
\Tbo{G_x}{G_y}{G_z} = \begin{pmatrix}
  J^{-T} & \begin{pmatrix} 0 \\ 0 \end{pmatrix} \\[2ex]
  \begin{pmatrix} 0 & 0 \end{pmatrix} & 1
\end{pmatrix} \Tbo{\partial_s{\bm p} \cdot {\bm g}}{\partial_n{\bm p} \cdot {\bm g}}{{\mathcal N}_P \cdot {\bm g}}. \label{eq:GravF}
\ee

\section{Single-track vehicle model} \label{sec:singletrack}
For pedagogical purposes, we will use a relatively simple single-track car model (sometimes known as the bicycle model), which has no roll freedom.
Models of this type have their genesis in \cite{2012-PacejkaBook}, and are analysed in detail in \cite{2018-Book-LimebeerMassaro}.
The vehicle frame is modelled as a rigid body; spinning road wheels are not represented.
The kinematics of the single-track vehicle model are shown in Figure\,\ref{fig:Bicycle}, where its ground contact points are on a plane tangent to the road surface.

\begin{figure}[ht!]
\begin{center}
  \includegraphics[trim=0mm 15cm 0mm 6cm, clip, width=.8\textwidth]{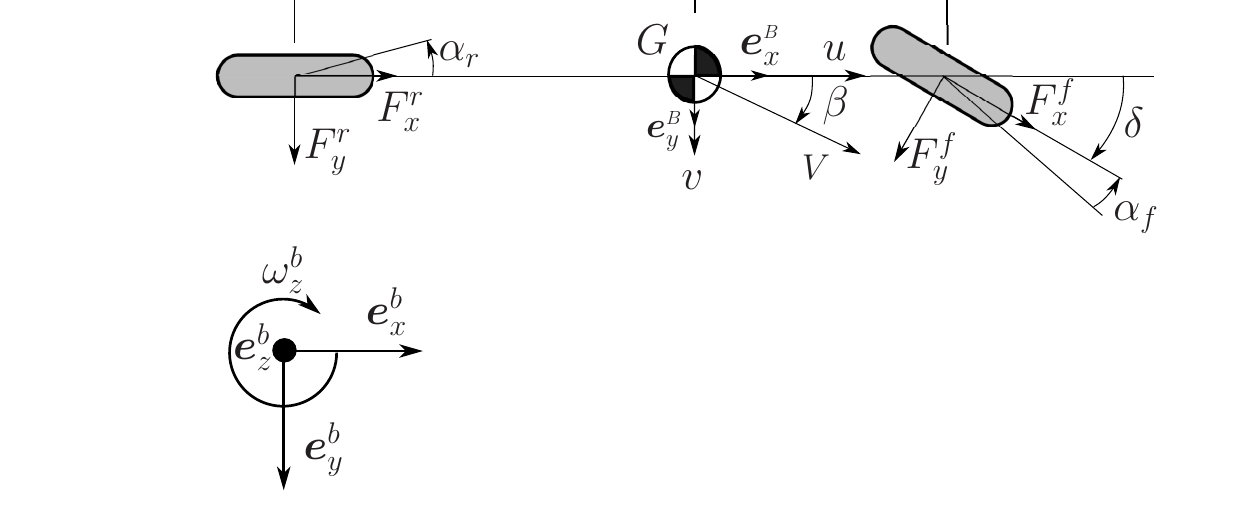}
  \end{center}
 \caption{Kinematics of a single-track car model showing its basic geometric parameters; ${\bm e}_z^b$ points into the page, with angles defined by a right-hand rule.}
 \label{fig:Bicycle}
\end{figure}

The tyre slip angles are given by
\be
\alpha_f = \mbox{atan2}(v+\omega_z^b a,u) - \delta \mbox{~~and~~}
\alpha_r = \mbox{atan2}(v-\omega_z^b b,u).
\ee
The vehicle's velocity components in the ground-plane curvilinear coordinate system
are dictated by the road surface geometry, and are given by \eqref{eq:LinEqnA}.
The vehicle's angular velocity and angular acceleration are given by \eqref{eq:AngVel} and \eqref{eq:AngAcc}, respectively.

\subsection{Vehicle dynamics} The equations of motion are based on the Newton-Euler equations; vector and tensor quantities are expressed in bold face characters:
\be
{\bm F}^b = m \left( \dot{\bm v}^b + {\bm \omega}^b \times {\bm v}^b \right) \mbox{~~and~~}
{\bm M}^b = \left({\bm I}_b \dot{\bm \omega}^b + {\bm \omega}^b \times ({\bm I}_b {\bm \omega}^b) \right),
\ee
where ${\bm v}^b$ and ${\bm \omega}^b$ are the vehicle's velocity and angular velocity vectors, respectively, expressed in body-fixed coordinates.
The mass of the vehicle is $m$, ${\bm I}_b$ the car's moment of inertia tensor expressed in body-fixed coordinates,
with ${\bm F}^b$ and ${\bm M}^b$ the external force and moment acting on the body (also expressed in body-fixed coordinates).
Since the car must remain in contact with the road, one of its translational freedoms (heave) and two of its rotational freedoms (pitch and roll) must be constrained.

The two translational and one rotational freedoms of the car are described by
\begin{eqnarray}
\dot{u} &=& (v+h\omega^b_x)\omega^b_z  + h \dot{\omega}^b_y  + F_x /m; \label{eq:u}\\% Equation (7.113)
\dot{v} &=& (h \omega^b_y - u) \omega^b_z - h \dot{\omega}^b_x + F_y/m;  \label{eq:v} \\  % Equation (7.114)
I_z \dot{\omega}^b_z &=& M_z \label{eq:Mz}  
\end{eqnarray}
in which $I_z$ is the car's yaw moment of inertia (expressed in body-fixed coordinates). The height of the vehicle's mass centre above the road is $h$.

The external forces acting on the car in a car-fixed coordinate system are given by
\begin{eqnarray}
F_x &=& F_{fx} \cos \delta - F_{fy} \sin \delta + F_{rx} + G_x + F_x^a; \label{long_in_f} \\
F_y &=& F_{fy} \cos \delta + F_{fx} \sin \delta + F_{ry} + G_y + F_y^a; \label{lat_in_f} \\
F_z &=& F_{rz} + F_{fz} + G_z + F_z^a, \label{eq:load_force_balance}
\end{eqnarray}
where $\delta$ is the car's steering angle, $F_{fx}$ denotes the $x$-component of the front tyre force; 
the other tyre force components are described in the same way\textemdash the $y$ and $z$ components are given subscripts $y$ and $z$. 
The aerodynamic forces and yaw moments are given by 
\beaw
F_x^a &= k C_D, ~
F_y^a = k C_S, ~
F_z^a = k C_L, ~ \nonumber \\
M_y^a &= k M_A C_y^M,~
M_z^a = k M_A C_z^M,
\eeaw
where $k=\frac{\rho}{2} A u^2$. The density of air is $\rho$,
$A$ its frontal area, $M_A$ is the yaw moment arm, and $C_D$, $C_S$, $C_L$, $C_y^M$ and $C_z^M$ are the drag, side force, lift, and pitch- and yaw-moment coefficients, respectively.
The gravitational force components in body-fixed axes are $G_x$, $G_y$ and $G_z$; see \eqref{eq:GravF}.
The tyre forces are described using empirical formulae; see Appendix\,A in \cite{Perantoni_2014}. These models produce the longitudinal and lateral tyre forces $F_x$ and ${F_y}$,
as functions of the tyres' normal load, longitudinal slip, and side-slip angle.

The external moments acting on the car are given by
\beaw
M_y &= b F_{rz} - a F_{fz} +h (F_{fx} \cos \delta - F_{fy} \sin \delta + F_{rx}) + M^a_y \label{ymom} \\ 
M_z &= a (F_{fy} \cos\delta + F_{fx} \sin\delta) - b F_{ry} + M_z^a. \label{zmom}
\eeaw
The yaw angular velocity comes from \eqref{eq:Mz}.
The vertical force balance constraint is given by
\be
0 = u \omega^b_y - v \omega^b_x -h((\omega^b_y)^2+(\omega^b_x)^2) + F_z/m,  \label{eq:w}
\ee
while the road-constrained pitch angular velocity must satisfy
\be
I_y \dot{\omega}_y^b = M_y + I_z \omega^b_z \omega^b_x. \label{eq:My} % Equation (7.131)
\ee
In this model, the vehicle has no roll freedom. Equation \eqref{eq:My} assumes a narrow-body configuration where the roll moment of inertia is negligible (\(I_x \ll I_z\)), which simplifies the gyroscopic coupling torque induced by the road-enforced angular velocity.
In this rigid-chassis formulation, the normal tyre forces $F_{fz}$ and 
$F_{rz}$ act as algebraic constraint forces rather than independent states. 
Since the road geometry completely prescribes the heave and pitch 
accelerations via \eqref{eq:w} and \eqref{eq:My}, these normal loads are 
algebraically determined at each time step by solving the coupled linear 
system formed by the vertical force balance \eqref{eq:load_force_balance} 
and the pitch moment constraint \eqref{eq:My}.

The heading angle is $\chi$, which is the angle between ${\bm e}_x^b$ and $\partial_s {\bm x}_P$,
which can be computed using \cite{LimebeerVSD2021}
{\small \be
\dot{\chi} = \omega_z^b + \frac{(\partial_{ss} {\bm x}_P \times \partial_s {\bm x}_P) \cdot {\mathcal N}_P}{\partial_s {\bm x}_P \cdot \partial_s {\bm x}_P} \dot{s}
+ \frac{(\partial_{sn} {\bm x}_P \times \partial_s {\bm x}_P) \cdot {\mathcal N}_P}{\partial_s {\bm x}_P \cdot \partial_s {\bm x}_P} \dot{n}; \label{eq:chiAng}
\ee}
see Figure\,\ref{fig:Angles}.
The yaw angular velocity $\omega_z^b$ comes from a yaw-moment balance equation given in \eqref{eq:Mz}.

\section{Optimal control} \label{sec:OCprob}
There is an expansive literature on the optimal control problem (OCP) of road vehicles, with some of this material reviewed
in \cite{2018-Book-LimebeerMassaro,MassaroLimebeer2021}. Our purpose here is to demonstrate how road-surface modelling, vehicle dynamics, and numerical optimal control can be brought together into a unified framework.

\subsection{States, inputs and constraints}
The control inputs to the system description are
\be
{\bm u} =[\dot{\delta},\dot{k}_{f},\dot{k}_{r},F_{fz},F_{rz},\dot{u},\dot{v}]^T, \label{eq:inputsU}
\ee
in which, $\dot{\delta}$ is the steering angular velocity, $\dot{k}_f$ and $\dot{k}_r$ are the rates of change of the tyres' longitudinal slips,
$F_{fz}$ and $F_{rz}$ are the tyre normal loads, while $\dot{u}$ and $\dot{v}$ are the vehicle's longitudinal and lateral acceleration components\textemdash 
the reason for the inclusion of the acceleration terms as inputs will be explained shortly.
 
In the case of the bicycle model described above, the vehicle states are
\be
{\mathbf x} = [n,\,\chi,\,v,\,\omega_z^b,\,u,\,\delta,\,k_f,\,k_r,\,t]^T. \label{eq:states}
\ee
The car's lateral position $n$ (with respect to the track centreline) is computed using the differential equation \eqref{eq:LinEqnA}.
The vehicle's heading angle is computed with the help of \eqref{eq:chiAng}.
The car's lateral velocity $v$ is given by \eqref{eq:v}.
The car's yaw velocity $\omega_z^b$ comes from \eqref{eq:Mz}.
The car's longitudinal velocity $u$ is given by \eqref{eq:u}.
The steering angle $\delta$ and the tyre longitudinal slip are obtained by integrating their corresponding inputs.
The front-wheel slip $k_f$  is constrained so that only braking on the front wheel is allowed.
The elapsed time comes from $t = \int \frac{dt}{ds} ds$ provides the lap time that is useful in a minimum lap time OCP.
The car's sideslip angle is given by $\beta=\arctan \frac{v}{u}$.
For convenience, the independent variable used in the OCP is the elapsed distance $s$ rather than the more conventional time variable.

\subsection{Index-1 implicit ODEs in a collocation framework}
Equations \eqref{eq:u}, \eqref{eq:v} and \eqref{eq:My} are index-1 implicit ODEs. 
A common approach to the removal of this difficulty is to
assume away the implicit nature of these equations by setting $\dot{\omega}^b_y = 0$ and $\dot{\omega}^b_x = 0$.
Alternatively, $\dot{u}$ and $\dot{v}$ are introduced as inputs that are then integrated to generate $u$ and $v$.
With these quantities available, \eqref{eq:LinEqnA}, \eqref{eq:AngVel} and \eqref{eq:AngAcc} can be used to generate
$\dot s$, $\dot n$, $\omega_x^b$, $\omega_y^b$, $\dot{\omega}_x^b$ and $\dot{\omega}_y^b$.
It is then possible to treat \eqref{eq:u}, \eqref{eq:v} and \eqref{eq:My} as equality constraints, 
rather than as explicit differential equations.

\subsection{Other constraints}
In order to keep the vehicle within the track boundaries, the
car's lateral positions are constrained by 
\begin{eqnarray}
r_w/2 \mp n  \mp (a \sin \chi \pm \frac{w}{2} \cos \chi) \cos \phi &\ge& 0 \label{eq:B1} \\
%r_w/2 + n  + (a \sin \chi - \frac{w}{2} \cos \chi) \cos \phi  &\ge& 0 \label{eq:B2} \\
r_w/2 \mp n  \pm (b \sin \chi \mp \frac{w}{2} \cos \chi) \cos \phi  &\ge& 0 \label{eq:B2}
%r_w/2 + n  - (b \sin \chi + \frac{w}{2} \cos \chi) \cos \phi &\ge& 0, \label{eq:B4}
\end{eqnarray}
in which $w$ is the car's notional track width.
The constraints \eqref{eq:B1} ensure that the front wheel remains within the track boundaries,
while \eqref{eq:B2} ensure that the rear wheel remains within the track boundaries.
The ground-plane projection of the track width is $r_w(s)$ and is estimated from measured LiDAR data.
The track's camber angle is $\phi$ which can be obtained from the track data \cite{Limebeer2025}, with $\chi$ described by \eqref{eq:chiAng}.
The vertical force is constrained using \eqref{eq:w}, while the pitch moment is constrained using \eqref{eq:My}.
Finally, the engine power can be constrained using $P_{\scriptscriptstyle\mathrm{max}} -u F_{rx} \ge 0$.
In the context of mechanics, $u$ is the longitudinal body-fixed velocity component,
while in control, ${\bm u}$ is the control input. This distinction should be clear from the context.
Since we are looking at an exemplar minimum lap time problem,
cyclicity constraints are used to enforce continuity in the states across the start-finish line.

\subsection{Performance index}
The performance index is given by
\be
J = \oint \frac{1}{\dot s} \left(1 + \sum_{i=1}^3 R_i u_i^2 \right) ds, \label{eq:PI1E}
\ee
in which the $R_i$'s are constant regularization weighting terms. 
The idea is that $\oint \frac{ds}{\dot s}$ is the lap time to be minimized, with only minor contributions attributable to the inputs.
Since time is a state, the true lap time is measured explicitly.
The regularization terms are used to dampen oscillatory responses and avoid singular arcs in the OCP's solution.

When structuring and solving problems of this type, several issues must be considered.
\begin{itemize}
\item[-] {\bf Nonlinear solver:} The selection of a numerical integration algorithm and an associated nonlinear programme solver. 
Here we use a direct numerical method based on orthogonal collocation.
These techniques are embedded in the software package
%GPOPSII \cite{Darby2,Patterson2015},
GPOPSII \cite{Patterson2015},
which makes use of the nonlinear programme solver `Interior Point OPTimizer' IPOPT \cite{ipopt}.
\item[-] {\bf Scaling:} We normalise the length of the vehicle, the mass of the vehicle, and scale time so that the gravitational acceleration of unity. All other units are based on the units of length, mass, and time.
\item[-] {\bf Nonsmooth functions:} The approximation of non-smooth functions such as $\max(\cdot,\cdot)$, $\min(\cdot,\cdot)$, and $|(\cdot)|$. As an example, $\max(a,b) = a + \max(0,b-a) \approx \frac{(b-a) + \sqrt{(b-a)^2 +\epsilon}}{2}$; other examples can be found in \cite{2018-Book-LimebeerMassaro}.
\item[-] {\bf Gradients:} The efficient computation of gradients. Here we used the MATLAB automatic differentiation package {\em ADiGator} that transforms user function files into
derivative function files \cite{2017-Adigator}.
%{\bf Adaptive mesh refinement strategies:} Here, we use the {\em ph} mesh-refinement algorithm described in \cite{Patterson2014}.
\item[-] {\bf Adaptive mesh refinement strategies:} Here, we use the {\em ph} mesh-refinement algorithm that allows for changes in both the number of mesh intervals and the degree of the approximating polynomial within a mesh interval.
\end{itemize}

\subsection{An example}
To evaluate the numerical performance of the fully integrated optimal control framework, a minimum-lap-time problem was executed 
on the Darlington Raceway, which has a total length of $2100$~m. The interested reader can find a detailed high-fidelity study using a two-track vehicle model with posture-dependent aerodynamics in \cite{LimebeerVSD2021}. Here, the simplified
single-track vehicle equations given in Section\,\ref{sec:singletrack} are used. 
The resulting optimal profiles are presented in Figure\,\ref{fig:ocp_results}.

\begin{figure}[b!]
\centering
\subfloat[Optimal vehicle speed profile.\label{fig:speed_profile}]{%
  \includegraphics[width=0.48\linewidth]{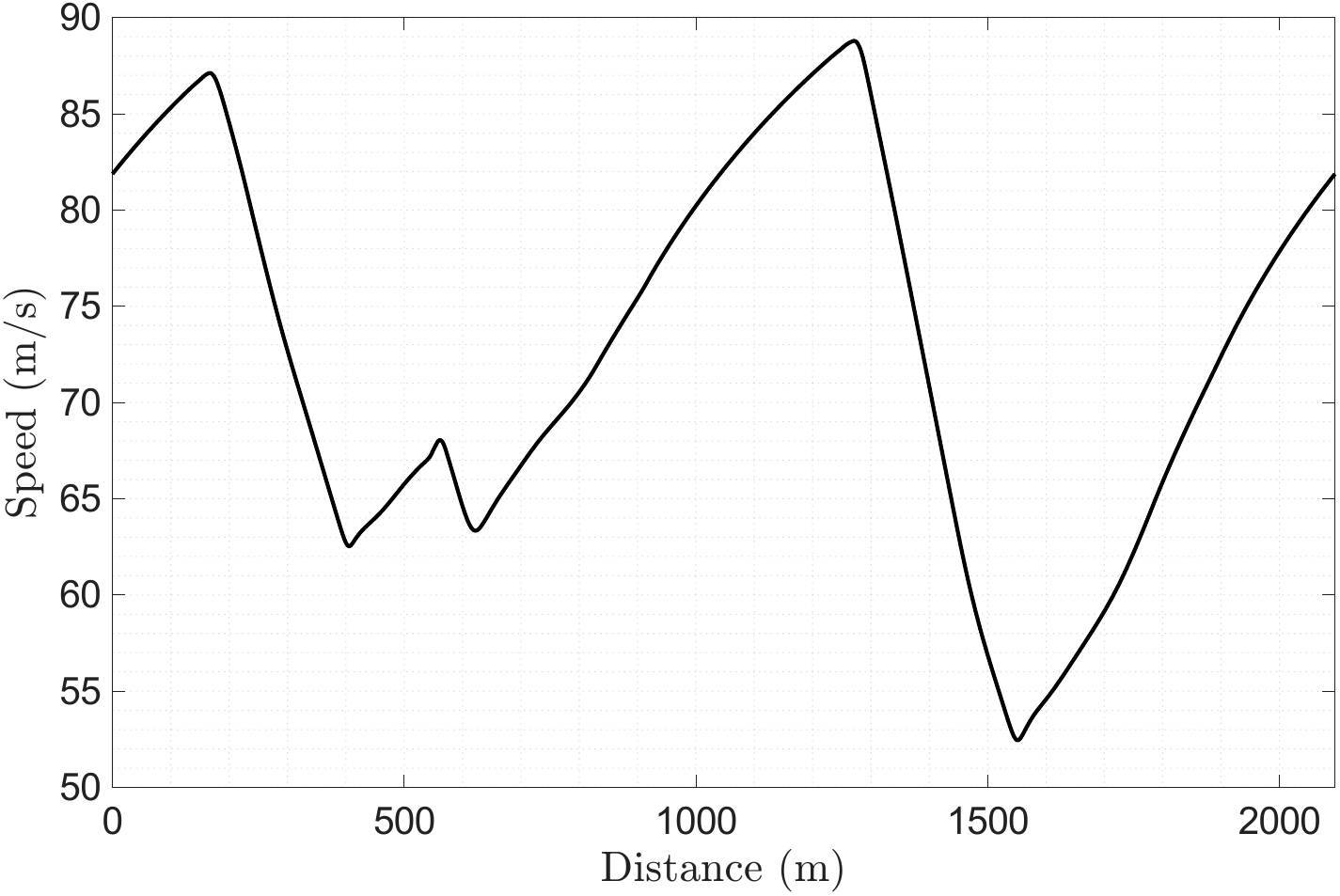}%
}
\hfill
\subfloat[Front ($\kappa_f$, blue) and rear ($\kappa_r$, red) tyre slip parameters.\label{fig:tyre_slips}]{%
  \includegraphics[width=0.48\linewidth]{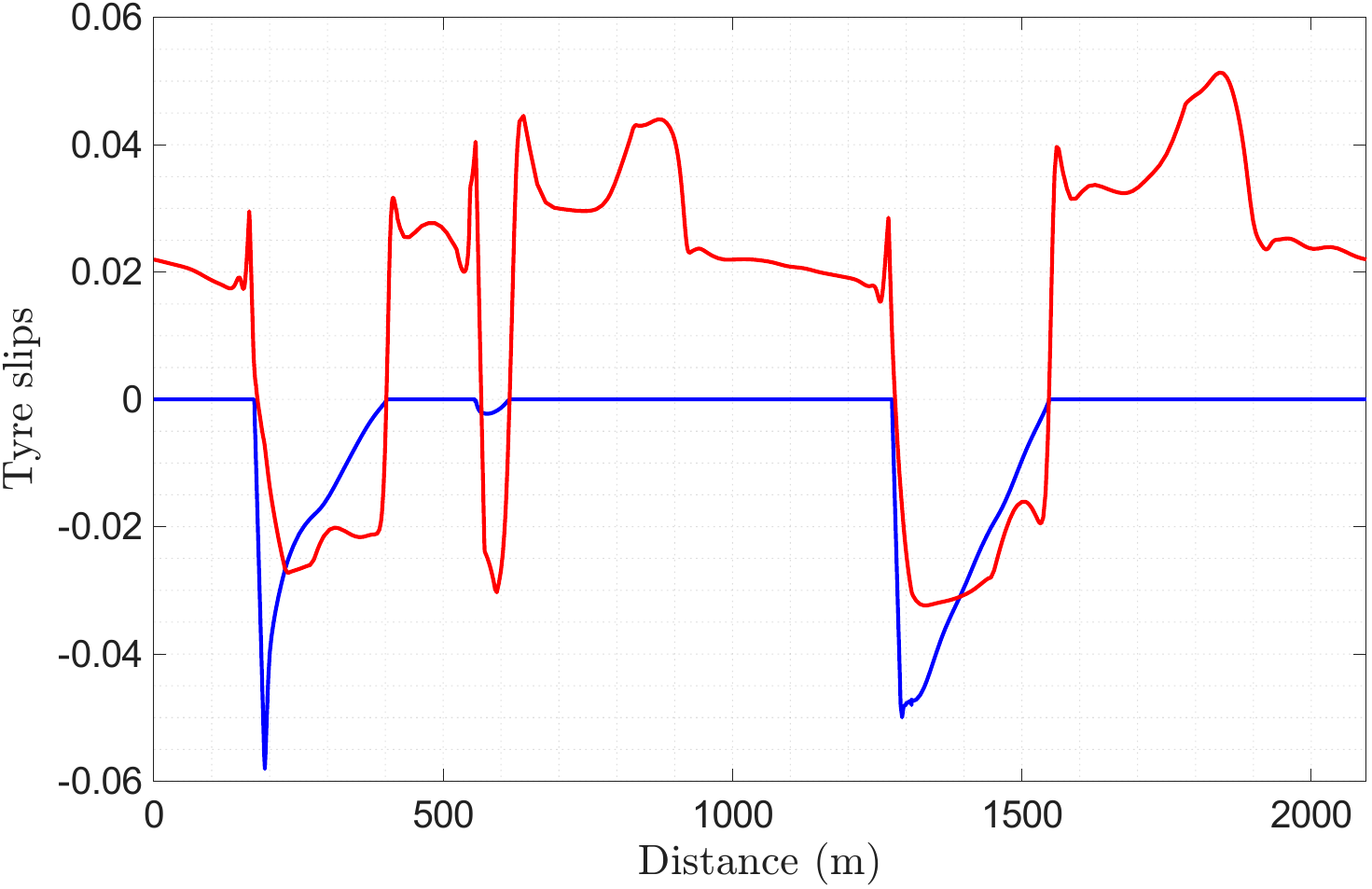}%
}
\caption{Trajectory optimization results computed over the $2100$~m Darlington Raceway, illustrating the explicit physical coupling between transient non-planar braking transitions and saturation limits of the tyre contact patches.}
\label{fig:ocp_results}
\end{figure}

The velocity trace (Figure\,\ref{fig:speed_profile}) shows the speed profile one would expect on a NASCAR oval.
The tyre slips shown in (Figure\,\ref{fig:tyre_slips}) reveal two clear braking zones, with the drive applied to the rear axle only.
The vehicle accelerates from a speed of $82$~m/s, as it crosses the start-finish line, to a speed of $87.1$\,m/s at $170$\,m.
At this point the vehicle enters a phase of heavy braking as it enters the first corner, which is accompanied by sharp negative spikes
in the longitudinal slips on both tyres. The front braking slip of $-0.056$ approaches the saturation limit.
This balances against a corresponding drop in rear tyre slip to $-0.026$.
As the vehicle drives through the section between $410$~m and $620$~m, front braking slip drops to zero, while the rear tyre transitions into full-power traction operating at a high slip ratio peak of nearly $0.035$ to maximize force generation during apex exit. 
Along the main straightaway from $620$~m to $1270$~m, steady acceleration carries the vehicle to its global maximum velocity of $88.7$~m/s, where rear longitudinal slip settles uniformly around a value of $0.022$. At $1270$~m, the vehicle undergoes a second braking phase into Turn\,3, establishing a global minimum velocity of $52.5$~m/s at $1550$~m.
This is accompanied by a second phase of heavy front-axle braking with the associated tyre slip dropping to $-0.05$. 
Finally, the vehicle exits the turn, pushing the rear tyre slip up to a peak of $0.051$ at approximately $1850$~m to maximize exit speed out to the final corner.
These non-smooth matched state transitions confirm that the $hp$-adaptive collocation algorithm successfully resolves the friction saturation boundaries under surface-related differential-geometric constraints.

\section{Conclusions} \label{sec:conclusions}
This article provides an expository overview of classical results in differential geometry essential for analyzing vehicle dynamics and optimal trajectories on curved road surfaces. The transformations established in \eqref{eq:LinEqnA}, \eqref{eq:AngVel}, \eqref{eq:AngAcc} and \eqref{eq:GravF} demonstrate that tyre forces, aerodynamic loads, and rigid-body mechanics are best expressed in body-fixed coordinates, whereas the underlying road geometry is most conveniently mapped via curvilinear track coordinates. Classical flat-surface models are readily recovered by setting the surface metric tensor to the Kronecker delta $\delta_{ij}$ and the shape operator to zero. 

While several comprehensive mathematical texts cover classical differential geometry, Shifrin~\cite{Shifrin2006} offers an accessible entry point that balances mathematical rigor with engineering utility. In order to appreciate the full physical picture, these differential-geometric constraints must be integrated directly into multibody vehicle mechanics models. The numerical trajectory results presented over the Darlington Raceway terrain model demonstrate that road curvature, position-dependent variations in gravitational forcing, and normal force fluctuations generate dynamically significant variations in local tyre grip capacities that cannot be routinely ignored in high-performance or autonomous vehicle simulations. While simple 3D topologies, such as purely undulating straight roads, may be adequate in highly restricted contexts, they are far too specialized for the demands of Formula One and NASCAR track modelling. The differential-geometric mapping engine presented here is general and mathematically robust enough to accommodate advanced multi-variable constraints, paving the way for future studies focusing on joint trajectory and hybrid powertrain power flow optimization for the current 2026 Formula One technical regulations.

\section*{Data Availability Statement}
The authors confirm that vehicle and tyre parameters, and pseudospectral optimization mesh criteria necessary to replicate the simulations are detailed in Appendix\,\ref{App:C}. The complete underlying MATLAB optimization scripts, track-modelling algorithms, and high-density LiDAR road surface profile data for Darlington Raceway are available from the corresponding author upon reasonable request. A parametric track-surface model was developed using the methods described in \cite{Limebeer2025}.

%%%%%%%%%%%%%%%%%%%%%%%%%%%%%%%%%%%%%%%%%%%%%%%%%%%%%%%%%%%%%%%%%%%%%%

%\twocolumn
\bibliographystyle{asmejour}   %% .bst file that follows ASME journal format. Do not change.
\bibliography{References} %% <=== change this to name of your bib file

\appendix
\section{Computation of a geodesic curve} \label{App:Geo}
MATLAB code used to compute a geodesic curve between points $(s_A,\,n_A)$ and $(s_B,\,n_B)$.
%\begin{center}
\begin{lstlisting}
%% Appendix A. Computation of a Geodesic Curve
% MATLAB code used to compute a geodesic curve between points (sA, nA) and (sB, nB).
%% Elliptic cone geodesics (ruled / developed representation)
clearvars;
close all;
clc;
%% 1. Geometric & Boundary Parameters
% Define the spatial dimensions of the elliptic cone and the coordinates 
% of the target points in the local orthogonal (s, n) coordinate system.
a   = 1.5;
b   = 1.0;
c   = 0.9;             % Cone geometry parameters (semi-axes and height scaling)
sA  = 1.0;  nA = 1.0;  % Geodesic start point coordinates on the manifold
sB  = -1.0; nB = -1.4; % Geodesic end point coordinates on the manifold

%% 2. Mapping theta -> phi via the Developability Condition
% The elliptic cone is a developable surface (Gaussian curvature K = 0 everywhere 
% except at the vertex). It can be unrolled into a flat Euclidean plane without 
% stretching. This mapping requires integrating the differential relationship 
% dphi = (sqrt(EG - F^2) / E) dtheta, which preserves isometric lengths.
% This maps the physical azimuthal angle (theta) to the flattened plane angle (phi).
% See Eq. (53) in the main text for the analytic derivation.
integrand = @(t) sqrt(a^2*b^2 + c^2*(a^2*sin(t).^2 + b^2*cos(t).^2)) ...
    ./ (a^2*cos(t).^2 + b^2*sin(t).^2 + c^2);
n_pre       = 5000; % Resolution for precomputing the transcendental mapping array
theta_pre   = linspace(0, 2*pi, n_pre);
phi_pre     = zeros(size(theta_pre));
for i       = 1:n_pre
    phi_pre(i) = integral(integrand, 0, theta_pre(i), 'AbsTol', 1e-12);
end
phi_total   = phi_pre(end);    % Total angular span of the unfolded cone sector in 2D
deficit     = 2*pi - phi_total; % Angular deficit characterizing the vertex singularity

fprintf('========== CONE PROPERTIES ==========\n');
fprintf('a = %.2f, b = %.2f, c = %.2f\n', a, b, c);
fprintf('Total flat angle = %.1f deg\n', rad2deg(phi_total));
fprintf('Deficit angle = %.1f deg\n', rad2deg(deficit));

% Construct function handles for fast bi-directional interpolation.
% This maintains mappings strictly within the principal topological interval [0, 2*pi].
phi_from_theta = @(theta) interp1(theta_pre, phi_pre, theta, 'linear', 0);
theta_from_phi = @(phi) interp1(phi_pre, theta_pre, phi, 'linear', 0);

%% 3. Geodesic Computation on the Flattened Pattern
% By definition, a geodesic (shortest path) on a developed surface unrolls 
% into a simple straight line in the 2D flattened coordinate space (U, V).
% First, transform the 3D boundary coordinates into the 2D isometric plane:
[UA, VA] = flatten_sn(sA, nA, phi_from_theta, a, b, c);
[UB, VB] = flatten_sn(sB, nB, phi_from_theta, a, b, c);

% Generate a linear parameterized tracking array connecting the unrolled points:
n_geo   = 300;
t       = linspace(0, 1, n_geo).';
U_geo   = UA + t*(UB - UA);
V_geo   = VA + t*(VB - VA);

% Pre-allocate coordinate arrays to map the flat straight line back onto the 3D surface:
s_geo       = zeros(size(U_geo));
n_geo_arr   = zeros(size(U_geo));
z_geo       = zeros(size(U_geo));

for i = 1:numel(U_geo)
    % Inverse-map the isometric 2D points (U, V) back to curvilinear (s, n) coordinates:
    [s_geo(i), n_geo_arr(i)] = unflatten_point(...
        U_geo(i), V_geo(i), theta_from_phi, phi_total, a, b, c);
    % Calculate the extrinsic vertical coordinate z from the surface constraint equation:
    z_geo(i) = c * sqrt(s_geo(i)^2/a^2 + n_geo_arr(i)^2/b^2);
end
% The true spatial geodesic distance equals the Euclidean distance in the isometric unrolled plane.
fprintf('\nGeodesic distance = %.4f\n', hypot(UB - UA, VB - VA));

%% 4. Graphical Layout & Reference Rendering Setup
s_max = 2.5;
n_max = 2.5;
R_max = min(s_max/a, n_max/b); % Establish scaling boundary for visualization

% Trace the continuous exterior perimeter of the flattened pattern boundary:
n_boundary  = 2000;
theta_b     = linspace(0, 2*pi, n_boundary);
[U_b, V_b]  = flatten_curve(theta_b, R_max, phi_from_theta, a, b, c);

% Cut Edges: The cone must be sliced along a single ray (seam) to flatten it.
% The ray theta = 0 maps to phi = 0, while the ray theta = 2*pi maps to the bounding ray phi_total.
% These two lines are physically coincident on the 3D manifold but distinct in the 2D unrolled plane.
[U_cut1, V_cut1] = cut_edge(0, R_max, 0, a, b, c);
[U_cut2, V_cut2] = cut_edge(2*pi, R_max, phi_total, a, b, c);

% Map structural interior rulings (straight line generators radiating from the vertex)
n_rulings   = 24;
theta_r     = linspace(0, 2*pi, n_rulings + 1);
theta_r     = theta_r(2:end-1); % Exclude the overlapping cut boundaries
[U_r, V_r]  = flatten_curve(theta_r, R_max, phi_from_theta, a, b, c);

s_r3d = a * R_max * cos(theta_r);
n_r3d = b * R_max * sin(theta_r);
z_r3d = c * R_max * ones(size(theta_r));

% Generate horizontal cross-sections (conic ellipses perpendicular to the main cone axis)
n_horiz     = 5;
R_horiz     = linspace(R_max/3, R_max, n_horiz);
theta_h     = linspace(0, 2*pi, 300);
U_h_cell    = cell(n_horiz, 1);   V_h_cell = cell(n_horiz, 1);
s_h_cell    = cell(n_horiz, 1);   n_h_cell = cell(n_horiz, 1);   z_h_cell = cell(n_horiz, 1);

for k = 1:n_horiz
    Rk = R_horiz(k);
    [U_h_cell{k}, V_h_cell{k}] = flatten_curve(theta_h, Rk, phi_from_theta, a, b, c);
    s_h_cell{k} = a * Rk * cos(theta_h);
    n_h_cell{k} = b * Rk * sin(theta_h);
    z_h_cell{k} = c * Rk * ones(size(theta_h));
end

%% 5. Create Figure Windows
figure('Position', [50, 50, 1400, 550]);
clf;

% --- LEFT PANEL: 3D Surface Rendering ---
subplot(1, 2, 1);
hold on;
[R_mesh, Th_mesh] = meshgrid(linspace(0, R_max, 40), linspace(0, 2*pi, 60));
S_mesh = a * R_mesh .* cos(Th_mesh);
N_mesh = b * R_mesh .* sin(Th_mesh);
Z_mesh = c * R_mesh;

surf(S_mesh, N_mesh, Z_mesh, ...
    'FaceAlpha', 0.25, 'EdgeColor', 'k', ...
    'EdgeAlpha', 0.15, 'FaceColor', [0.6 0.7 0.9]);

% Superimpose geometric components on the 3D cone
for i = 1:numel(theta_r)
    plot3([0, s_r3d(i)], [0, n_r3d(i)], [0, z_r3d(i)], 'k-', 'LineWidth', 1);
end
plot3([0, a*R_max], [0, 0], [0, c*R_max], 'g--', 'LineWidth', 2); % The physical cut line seam

for k = 1:n_horiz
    plot3(s_h_cell{k}, n_h_cell{k}, z_h_cell{k}, 'r-', 'LineWidth', 1);
end

% Plot the calculated spatial geodesic path and target points
plot3(s_geo, n_geo_arr, z_geo, 'b-', 'LineWidth', 2.5);
zA = c * sqrt(sA^2/a^2 + nA^2/b^2);
zB = c * sqrt(sB^2/a^2 + nB^2/b^2);
plot3(sA, nA, zA, 'bo', 'MarkerSize', 8, 'MarkerFaceColor', 'c');
plot3(sB, nB, zB, 'ro', 'MarkerSize', 8, 'MarkerFaceColor', 'r');
plot3(0, 0, 0, 'ko', 'MarkerSize', 6, 'MarkerFaceColor', 'k');
hold off;

xlabel('$s$', 'Interpreter', 'latex', 'FontSize', 20);
ylabel('$n$', 'Interpreter', 'latex', 'FontSize', 20);
zlabel('$z$', 'Interpreter', 'latex', 'FontSize', 20);
view(65, 30); axis equal; grid minor; box on;
xlim([-s_max, s_max]); ylim([-n_max, n_max]); zlim([0, c*R_max]);

% --- RIGHT PANEL: Flattened Pattern (Isometric Plane) ---
subplot(1, 2, 2);
hold on;
% Render the continuous bounding area of the unrolled sheet sector
fill([0, U_cut1, U_b, U_cut2, 0], [0, V_cut1, V_b, V_cut2, 0], ...
    [0.85 0.85 0.95], 'EdgeColor', 'none', 'FaceAlpha', 0.5);
plot(U_b, V_b, 'b-', 'LineWidth', 1.5);

for i = 1:numel(theta_r)
    plot([0, U_r(i)], [0, V_r(i)], 'k-', 'LineWidth', 0.7);
end

% Plot separated cut edges showing the split topology boundaries
plot([0, U_cut1], [0, V_cut1], 'k-', 'LineWidth', 1);
plot([0, U_cut1], [0, V_cut1], 'g--', 'LineWidth', 2);
plot([0, U_cut2], [0, V_cut2], 'k-', 'LineWidth', 1);
plot([0, U_cut2], [0, V_cut2], 'g--', 'LineWidth', 2);

for k = 1:n_horiz
    plot(U_h_cell{k}, V_h_cell{k}, 'r-', 'LineWidth', 1);
end

% Plot the isometric straight-line geodesic path segment
plot(U_geo, V_geo, 'c-', 'LineWidth', 2.5);
plot(UA, VA, 'bo', 'MarkerSize', 8, 'MarkerFaceColor', 'b');
plot(UB, VB, 'ro', 'MarkerSize', 8, 'MarkerFaceColor', 'r');
plot(0, 0, 'ko', 'MarkerSize', 6, 'MarkerFaceColor', 'k');
hold off;

xlabel('$U$', 'Interpreter', 'latex', 'FontSize', 18);
ylabel('$V$', 'Interpreter', 'latex', 'FontSize', 18);
axis equal; grid minor; box on;

%% 6. Local Transformation Functions
function [U, V] = flatten_curve(theta, R, phi_from_theta, a, b, c)
    % Maps curvilinear track boundaries directly into 2D isometric space.
    % The metric component E represents the localized scaling value (g_thetatheta)
    % mapping the metric stretching factor along the unrolling boundaries.
    phi = phi_from_theta(theta);
    E = a^2*cos(theta).^2 + b^2*sin(theta).^2 + c^2;
    r = R .* sqrt(E);
    U = r .* cos(phi);
    V = r .* sin(phi);
end

function [U, V] = cut_edge(theta, R, phi, a, b, c)
    % Handles isolated boundary mapping constraints along the cut rays without
    % triggering interpolation errors at the topological boundaries (0 and 2*pi).
    E = a^2*cos(theta)^2 + b^2*sin(theta)^2 + c^2;
    r = R * sqrt(E);
    U = r * cos(phi);
    V = r * sin(phi);
end

function [U, V] = flatten_sn(s, n, phi_from_theta, a, b, c)
    % Converts physical surface coordinates (s, n) into the intermediate configuration
    % parameters (R, theta) before projecting them into the isometric flat plane.
    R = sqrt(s^2/a^2 + n^2/b^2);
    theta = atan2(n/b, s/a);
    if theta < 0
        theta = theta + 2*pi;
    end
    [U, V] = flatten_curve(theta, R, phi_from_theta, a, b, c);
end

function [s, n] = unflatten_point(U, V, theta_from_phi, phi_total, a, b, c)
    % Maps flat coordinates back onto the 3D cone surface.
    % Utilizes modular mapping transformations to ensure points crossing the cut
    % edge wrap correctly around the periodic manifold.
    r   = hypot(U, V);
    phi = atan2(V, U);
    if phi < 0
        phi = phi + 2*pi;
    end
    phi     = mod(phi, phi_total);
    theta   = theta_from_phi(phi);
    E       = a^2*cos(theta)^2 + b^2*sin(theta)^2 + c^2;
    R       = r / sqrt(E);
    s       = a * R * cos(theta);
    n       = b * R * sin(theta);
end
\end{lstlisting}
%\end{center}

\section{Dynamics on an elliptic cone} \label{App:DynEllip}
MATLAB code used to compute the dynamics of a unit mass on the interior surface of an elliptic cone.

\begin{lstlisting}
%% Appendix B. Dynamics on an Elliptic Cone
% MATLAB code used to compute the dynamics of a unit mass on the interior surface of an elliptic cone.
%% Dynamics with gravity on elliptic cone z = c * sqrt((s/a)^2 + (n/b)^2)
clearvars;
close all;
clc;

%% 1. Configuration Space & Initial State Parameters
a = 1.5; b = 1.0; c = 1.0; % Surface geometry scaling factors
g_grav = 9.81;             % Acceleration due to gravity (SI units: m/s^2)

% Define initial state matrix: [s0; n0; s_dot0; n_dot0]
s0 = 1.0; n0 = 1.5; ds0 = 0.8; dn0 = 0.0;
tspan   = [0 25];            % Simulation temporal window (seconds)
y0      = [s0; n0; ds0; dn0];

%% 2. Numerical Integration Setup
% Evaluates the structural ordinary differential equations (ODEs) derived natively 
% from the Euler-Lagrange framework mapped across the Riemannian metric.
[t, y] = ode45(@(t,y) geodesic(t, y, a, b, c, g_grav), tspan, y0);

% Extract system coordinates and generalized velocity state paths
s   = y(:,1);  n = y(:,2);
ds  = y(:,3); dn = y(:,4);

%% 3. Technical Visualizations & Figures
% --- FIGURE 1: 3D Manifold Surface with Spatial Trajectory ---
figure(1); clf;
S = linspace(-2.5, 2.5, 30);
N = linspace(-2.5, 2.5, 30);
[Sg, Ng] = meshgrid(S, N);

% Reconstruct the extrinsic vertical constraint equation z = f(s,n)
Zg = c * sqrt((Sg/a).^2 + (Ng/b).^2);
surf(Sg, Ng, Zg, 'EdgeColor', 'k', 'FaceAlpha', 0.5); hold on;

% Project the integrated state path directly onto the 3D surface manifold
z_traj = c * sqrt((s/a).^2 + (n/b).^2);
plot3(s, n, z_traj, 'k-', 'LineWidth', 2);
plot3(s(1), n(1), z_traj(1), 'bo', 'MarkerSize', 8, 'MarkerFaceColor', 'b');     % Trajectory Start Anchor
plot3(s(end), n(end), z_traj(end), 'ro', 'MarkerSize', 8, 'MarkerFaceColor', 'r'); % Trajectory End Anchor

xlabel('$s$', 'Interpreter', 'latex', 'FontSize', 18);
ylabel('$n$', 'Interpreter', 'latex', 'FontSize', 18);
zlabel('$z$', 'Interpreter', 'latex', 'FontSize', 18);
view(60, 55); grid minor; box on;

% --- FIGURE 2: System Phase Portraits ---
figure(2); clf;
% Left Panel: s-coordinate Phase Portrait
subplot(1, 2, 1);
plot(s, ds, 'b-', 'LineWidth', 1.5); hold on;
plot(s(1), ds(1), 'bo', 'MarkerSize', 8, 'MarkerFaceColor', 'b');
plot(s(end), ds(end), 'ro', 'MarkerSize', 8, 'MarkerFaceColor', 'r');
xlabel('$s$', 'Interpreter', 'latex', 'FontSize', 18);
ylabel('$\dot{s}$', 'Interpreter', 'latex', 'FontSize', 18);
grid minor; box on;

% Right Panel: n-coordinate Phase Portrait
subplot(1, 2, 2);
plot(n, dn, 'r-', 'LineWidth', 1.5); hold on;
plot(n(1), dn(1), 'bo', 'MarkerSize', 8, 'MarkerFaceColor', 'b');
plot(n(end), dn(end), 'ro', 'MarkerSize', 8, 'MarkerFaceColor', 'r');
xlabel('$n$', 'Interpreter', 'latex', 'FontSize', 18);
ylabel('$\dot{n}$', 'Interpreter', 'latex', 'FontSize', 18);
grid minor; box on;

%% 4. Core ODE Geometric Mechanics Engine
function dydt = geodesic(~, y, a, b, c, g_grav)
    % Extract generalized positions (q) and velocities (q_dot)
    s = y(1);  n = y(2);
    ds = y(3); dn = y(4);
    
    % Radial parameter regularizing the apex to prevent coordinate singularity division errors
    r = max(sqrt((s/a)^2 + (n/b)^2), 1e-8); 
    
    %% 4.1 Surface Gradient Calculations
    % Compute partial tangents mapping the surface slope (z_s and z_n)
    fs = c*s / (a^2 * r);
    fn = c*n / (b^2 * r);
    
    %% 4.2 First Fundamental Form (Metric Tensor Matrix: [E, F; F, G])
    % Redefines the inner product rules natively on the smooth surface manifold.
    E = 1 + fs^2;
    F = fs * fn;
    G = 1 + fn^2;
    
    %% 4.3 Second-Order Curvature Sensitivities
    % Calculate analytic mixed partial derivatives using the chain rule
    dfs_ds =  c*n^2 / (a^2 * b^2 * r^3);
    dfs_dn = -c*s*n / (a^2 * b^2 * r^3);
    dfn_ds = dfs_dn; % Enforce Clairaut's theorem (symmetry of mixed partials)
    dfn_dn =  c*s^2 / (a^2 * b^2 * r^3);
    
    % Spatial gradients of the metric tensor elements w.r.t the local coordinates
    dE_ds = 2*fs*dfs_ds;            dE_dn = 2*fs*dfs_dn;
    dF_ds = dfs_ds*fn + fs*dfn_ds;  dF_dn = dfs_dn*fn + fs*dfn_dn;
    dG_ds = 2*fn*dfn_ds;            dG_dn = 2*fn*dfn_dn;
    
    %% 4.4 Christoffel Symbols computation (Connection Matrix Coefficient Groups)
    % The metric matrix is inverted and multiplied against the first-kind symbol array M
    % to yield the structural connection coefficients Gamma^k_ij (Symbols of the Second Kind).
    M = [0.5*dE_ds,         0.5*dE_dn,         dF_dn - 0.5*dG_ds;
         dF_ds - 0.5*dE_dn, 0.5*dG_ds,         0.5*dG_dn];
    
    Gamma = [E, F; F, G] \ M; % Efficient matrix left-division to raise the indices
    
    % Component mapping matching the structural geometric acceleration equations:
    Gs_ss = Gamma(1,1); Gs_sn = Gamma(1,2); Gs_nn = Gamma(1,3); % Coefficients driving s-acceleration
    Gn_ss = Gamma(2,1); Gn_sn = Gamma(2,2); Gn_nn = Gamma(2,3); % Coefficients driving n-acceleration
    
    %% 4.5 Covariant Gravitational Force Vector
    % The conservative potential energy V = m*g*z yields a covariant gradient vector.
    % Left-multiplying by the inverse metric raises the index to resolve contravariant forcing.
    force   = -[E, F; F, G] \ (g_grav * [fs; fn]);
    force_s = force(1);
    force_n = force(2);
    
    %% 4.6 Final Equations of Motion (Geometric Form)
    % Formulated as: q_ddot^k = - Gamma^k_ij * q_dot^i * q_dot^j + Force^k
    dds = -(Gs_ss*ds^2 + 2*Gs_sn*ds*dn + Gs_nn*dn^2) + force_s;
    ddn = -(Gn_ss*ds^2 + 2*Gn_sn*ds*dn + Gn_nn*dn^2) + force_n;
    
    dydt = [ds; dn; dds; ddn];
end
\end{lstlisting}

%%%%%%%%%%%%%%%%%%%%%%%%%%%%%%%%%%%%%%%%%%%%%%%%%%%%%%%%%%%%%%%%%%%%%%
\section{Numerical Optimization Framework and Vehicle Parameters} \label{App:C}
\label{app:simulation_parameters}
This appendix compiles the structural dimensions, coordinate scaling matrices, localized surface tyre constants, and discrete mesh tracking parameters implemented within the optimal control trajectory algorithms detailed in Section~8. 

The baseline mechanical reference properties and scaling transformation variables extracted directly from the execution scripts are compiled in Table~\ref{tab:car_parameters}. The load-dependent tyre scaling metrics mapped within the contact patch force models are explicitly detailed in Table~\ref{tab:tyre_parameters}, and the numerical configuration criteria governing the Radau Pseudospectral transcription process are detailed in Table~\ref{tab:gpops_parameters}.

%=====================================================================
% TABLE 1: CAR PARAMETERS
%=====================================================================
\begin{table}[htbp]
\caption{Rigid-body physical parameters and scaling definitions of the vehicle model}
\label{tab:car_parameters}
\centering
\small
\begin{tabular}{llll}
\hline
\textbf{Parameter Definition} & \textbf{Symbol} & \textbf{Physical Value} & \textbf{Non-Dimensional Expression} \\ \hline
Reference Wheelbase Base Length & $w_0$ & $2.79$~m & $w_0 \cdot \text{lengthscale} = 1.0$ \\
Total Vehicle Mass & $M_0$ & $1550$~kg & $M_0 \cdot \text{massscale} = 1.0$ \\
Acceleration Due to Gravity & $g_0$ & $9.81$~$\text{m/s}^2$ & $g_0 \cdot \text{accscale} = 1.0$ \\
Distance from CG to Front Axle & $a$ & $1.32$~m & $1.32 \cdot \text{lengthscale} \approx 0.4731$ \\
Distance from CG to Rear Axle & $b$ & $1.47$~m & $1.47 \cdot \text{lengthscale} \approx 0.5269$ \\
Center of Gravity Height & $h$ & $0.38$~m & $0.38 \cdot \text{lengthscale} \approx 0.1362$ \\
Pitch Moment of Inertia & $I_y$ & $1200$~$\text{kg}\cdot\text{m}^2$ & $1200 \cdot \text{lengthscale}^2 \cdot \text{massscale} \approx 0.1541$ \\
Yaw Moment of Inertia & $I_z$ & $2500$~$\text{kg}\cdot\text{m}^2$ & $2500 \cdot \text{lengthscale}^2 \cdot \text{massscale} \approx 0.3211$ \\
Drag Aerodynamic Coefficient & $C_D$ & $-0.469$ & $-0.469$ \\
Sideforce Aerodynamic Coefficient & $C_S$ & $-0.169$ & $-0.169$ \\
Lift Aerodynamic Coefficient & $C_L$ & $1.090$ & $1.090$ \\
Pitching Moment Coefficient & $C_{PM}$ & $-0.013$ & $-0.013$ \\
Yawing Moment Coefficient & $C_{YM}$ & $0.037$ & $0.037$ \\
Ambient Air Density & $\rho$ & $1.156$~$\text{kg/m}^3$ & $1.156 \cdot \text{massscale}/\text{lengthscale}^3 \approx 16.273$ \\
Frontal Projected Reference Area & $A_f$ & $2.24$~$\text{m}^2$ & $2.24 \cdot \text{lengthscale}^2 \approx 0.2878$ \\
Moment Arm Reference Metric & $M_{\text{arm}}$ & $2.794$~m & $2.794 \cdot \text{lengthscale} \approx 1.0014$ \\ \hline
\end{tabular}
\end{table}

%=====================================================================
% TABLE 2: TYRE PARAMETERS
%=====================================================================
\begin{table}[htbp]
\caption{Front and rear tyre patch performance parameter specifications \cite{Perantoni_2014}}
\label{tab:tyre_parameters}
\centering
\small
\begin{tabular}{lll}
\hline
\textbf{Parameter Definition} & \textbf{Front Tyre (Value)} & \textbf{Rear Tyre (Value)} \\ \hline
First Reference Normal Load, $F_{z1}$ (N)   & $-5184.31$ & $-6175.15$ \\
Second Reference Normal Load, $F_{z2}$ (N)  & $-14632.51$ & $-13979.55$ \\
Longitudinal Friction Peak 1, $\mu_{x\max1}$ (dimensionless) & $1.2178$ & $1.2002$ \\
Longitudinal Friction Peak 2, $\mu_{x\max2}$ (dimensionless) & $1.0486$ & $1.0604$ \\
Slip Ratio Limit at Peak 1, $\kappa_{\max1}$ (dimensionless)  & $0.0847$ & $0.0818$ \\
Slip Ratio Limit at Peak 2, $\kappa_{\max2}$ (dimensionless)  & $0.0595$ & $0.0609$ \\
Lateral Friction Peak 1, $\mu_{y\max1}$ (dimensionless)     & $1.1252$ & $1.1040$ \\
Lateral Friction Peak 2, $\mu_{y\max2}$ (dimensionless)     & $0.8581$ & $0.8767$ \\
Slip Angle Limit at Peak 1, $\alpha_{\max1}$ (rad)           & $0.0863$ & $0.0806$ \\
Slip Angle Limit at Peak 2, $\alpha_{\max2}$ (rad)           & $0.0480$ & $0.0585$ \\
Pacejka Shape Curvature Exponent, $Q_x$ (dimensionless)      & $1.4676$ & $1.4661$ \\
Pacejka Shape Curvature Exponent, $Q_y$ (dimensionless)      & $1.2865$ & $1.3618$ \\
Pacejka Scaling Factor, $S_x$ (dimensionless)                & $1.4111$ & $1.4116$ \\
Pacejka Scaling Factor, $S_y$ (dimensionless)                & $1.4518$ & $1.4310$ \\ \hline
\end{tabular}
\end{table}

%=====================================================================
% TABLE 3: GPOPS PARAMETERS
%=====================================================================
\begin{table}[htbp]
\caption{GPOPS-II/IPOPT numerical optimization configuration parameters}
\label{tab:gpops_parameters}
\centering
\small
\begin{tabular}{ll}
\hline
\textbf{Configuration Parameter} & \textbf{Value / Assigned Setting} \\ \hline
Transcription Method             & Radau Pseudospectral Method (RPM-Differentiation) \\
NLP Solver Block                 & IPOPT \\
Linear Solver Engine             & MA57 \\
Maximum Internal Iterations      & 2000 \\
Derivative Supplier Utility      & ADiGator (Algorithmic Differentiation) \\
Derivative Order Level           & Second-order analytical sensitivity tracking \\
Mesh Refinement Logic            & $hp$-PattersonRao adaptive strategy \\
Mesh Tolerance Convergence       & $1.0 \times 10^{-5}$ (maximum relative error) \\
Initial Mesh Segments ($N$)      & 20 identical localized intervals \\
Collocation Points Per Segment   & 4 points (uniform distribution array) \\
Point Boundaries ($N_{\min}, N_{\max}$) & Minimum: 4, Maximum: 15 collocation nodes \\
Max Refinement Passes            & 20 refinement iterations \\
Scaling Control Protocol         & Automatic-bounds variable mapping \\ \hline
\end{tabular}
\end{table}
\end{document}